\documentclass{article}
\usepackage[utf8]{inputenc}

\usepackage[hidelinks]{hyperref}
\usepackage{amsmath,amsfonts,amsthm,cleveref}
\usepackage[a4paper,top=2.9cm,bottom=2cm,left=2.5cm,right=2.5cm,marginparwidth=1.75cm]{geometry}
\usepackage{enumitem}
\usepackage{graphicx, subcaption}
\graphicspath{ {figures/} }
\usepackage{tikz}
\usepackage[sort&compress,numbers]{natbib}

\newcommand{\de}{\: \mathrm{d}}
\newcommand{\ddt}[2]{\frac{\mathrm{d}^2#1}{\mathrm{d}#2^2}}

\newcommand{\R}{\mathbb{R}}
\newcommand{\N}{\mathbb{N}}

\renewcommand{\c}{\overline{c}}
\newcommand{\tr}{\mathrm{tr}}
\newcommand{\SL}{\mathrm{SL}(2,\R)}

\newtheorem{theorem}{Theorem}[section]

\newtheorem{lemma}[theorem]{Lemma}

\newtheorem{definition}[theorem]{Definition}
\theoremstyle{definition}
\newtheorem{remark}[theorem]{Remark}

\title{Symmetry-induced quasicrystalline waveguides}
\author{Bryn Davies\thanks{Department of Mathematics, Imperial College London, South Kensington Campus, London, SW7~2AZ, UK (email: bryn.davies@imperial.ac.uk).} \and Richard V. Craster\footnotemark[1]}
\date{}

\begin{document}
\maketitle

\begin{abstract}
    Introducing an axis of reflectional symmetry in a quasicrystal leads to the creation of localised edge modes that can be used to build waveguides. We develop theory that characterises reflection-induced localised modes in materials that are formed by recursive tiling rules. This general theory treats a one-dimensional continuous differential model and describes a broad class of both quasicrystalline and periodic materials. We present an analysis of a material based on the Fibonacci sequence, which has previously been shown to have exotic, Cantor-like spectra with very wide spectral gaps. Our approach provides a way to create localised edge modes at frequencies within these spectral gaps, giving strong and stable wave localisation. We also use our general framework to make a comparison with reflection-induced modes in periodic materials. These comparisons show that while quasicrystalline waveguides enjoy enhanced robustness over periodic materials in certain settings, the benefits are less clear if the decay rates are matched. This shows the need to carefully consider equivalent structures when making robustness comparisons and to draw conclusions on a case-by-case basis, depending on the specific application.
\end{abstract}


\section{Introduction}

A widespread goal of wave physicists and engineers is to design and build devices that efficiently direct wave energy along specific paths. A common approach is to start with an inhomogeneous material that has a spectral gap (a range of frequencies that do not propagate through the material) and perturb the material by introducing a defect or interface; if done correctly, this perturbation will create \emph{edge modes} that are strongly localised along the path of the defect. The eigenfrequency of these edge modes typically falls within the spectral gaps of the original material \cite{asboth2016short, hoefer2011defect, craster2022ssh}. This strategy is typically employed for periodic materials, where the spectral gap structure is easy to understand using Floquet-Bloch theory \cite{kuchment1993floquet} and the technique has led to a rich variety of waveguides, capable of guiding waves robustly along intricate paths \cite{khanikaev2013photonic, rechtsman2013photonic, makwana2018geometrically}. In this work, we will show how to generalise this principle to non-periodic settings, aiming to realise the under utilised potential of such structures. In particular, the potential advantages of quasicrystalline waveguides include very large spectral gaps, even at high frequencies, \cite{zolla1998remarkable}, very broadband effects, high quality factors \cite{marti2022high} and robustness \cite{marti2021edge, marti2022high, liu2022topological, ni2019observation}.

The media studied here are generated recursively by tiling different materials according to tiling rules. The set of such patterns is very broad and has been studied extensively by mathematicians. They gained popular attention following the discovery of an aperiodic set of tiles by Berger in 1966 \cite{berger1966undecidability}, which disproved previous conjectures that if a set of tiles can tile the plane, then they can always be arranged to do so periodically \cite{wang1961proving}. Many such aperiodic tilings have since been discovered, most notably those composed of just two tiles that were discovered by Penrose in the 1970s \cite{penrose1974role}. Following the discovery of a metallic alloy with quasicrystalline atomic structure (by Dan Shechtman, who would win a Nobel prize for this work \cite{shechtman1984metallic}), the excitement spread beyond mathematics and the field of quasicrystals was born.

The main example we will use to demonstrate our results will be based on a Fibonacci tiling rule, where each new term in the sequence is formed by combining the two previous terms. The wave transmission properties of such materials have been studied extensively, such that their spectral gap structure is well understood \cite{kohmoto1983localization, kohmoto1987critical, hassouani2006surface, guenneau2008acoustic, gei2010wave, gei2018waves, gei2020phononic}. It has been shown that this structure has self similarity \cite{gei2018waves} and demonstrates a Cantor-like emergence of nested spectral gaps, converging to a fractal pattern \cite{gei2018waves}. These properties have even been demonstrated experimentally for a system of structured rods \cite{gei2020phononic}. Typically, these previous studies have used periodic approximants to understand the structure of the spectrum for increasing iterations in the Fibonacci sequence. These \emph{supercell} methods have the advantage that Floquet-Bloch theory can be used as the structure is periodic. As the size of the repeating supercell increases, the Floquet-Bloch spectrum becomes increasingly intricate, yielding the spectrum of the quasicrystal in the limit (this is shown nicely in Figure~7 of \cite{spurrier2020kane}).

A common feature of the previous studies of waves in Fibonacci quasicrystals, and also of the present work, is that the materials considered are one-dimensional. Some general principles for understanding the diffraction of one-dimensional two-material substitutional systems were established by \cite{kolar1993quasicrystals}. A convenient property of one-dimensional materials is that their properties can be described concisely using transfer matrices. A review of the use of transfer matrices for modelling wave propagation in one-dimensional lossless media, including quasicrystals, can be found in \cite{SANCHEZSOTO2012191}. To study quasicrystalline unbounded materials we will need to consider the limit of the product of infinitely many transfer matrices, ordered according to the relevant tiling rule. This is a subtle problem in general \cite{felbacq1998limit, zolla1998wave}, however in the specific case of the Fibonacci problem the traces of successive products of transfer matrices can be related by a second-order non-linear recursion relation, discovered by Kohmoto in 1983 \cite{kohmoto1983localization}, which will simplify the analysis.

It has long been known that quasicrystalline patterns can be obtained by projecting a non-commensurate slice of a higher-dimensional periodic pattern (known as a \emph{superspace}). If we are able to relate the properties of the two materials, then this offers a very promising approach for understanding the wave transmission properties of quasicrystals, since the properties of the periodic superspace material are straightforward to understand using Floquet-Bloch analysis. This was utilised by \cite{johnson2008quasicrystal} to predict the density of states in quasicrystalline materials, obtained via projection. Similarly, projection of periodic materials has been used to derive homogenised regimes for quasicrystalline materials \cite{bouchitte2010homogenization, wellander2018two}. There are outstanding questions related to understanding the significance of Bloch wavenumbers in the superspace material and it is not clear how to relate edge modes in the superspace and to those in the projected quasicrystal. We will not explore these important questions in this work, but highlight this as an appealing avenue for future study.

In this work, we will create localised modes by introducing an axis of reflectional symmetry to the material. This strategy has been used previously to create reflection-induced edge modes, for example in quasicrystalline arrays of point scatterers on elastic plates by \cite{marti2021edge} and, experimentally, in a quasicrystalline acoustic waveguide by \cite{apigo2019observation}. We will derive novel theory for characterising these modes, which emerge as eigenmodes of a continuous differential model. The model we shall consider will be one dimensional, as we will focus on the axis perpendicular to the axis of reflection in order to capture the decaying modes (see \Cref{fig:GuideSketch} for an example). The other type of localised edge modes that have been studied extensively are those that occur at the edges of a quasicrystalline material with finite dimensions \cite{xia2020topological, kraus2012topological, bandres2016topological, ni2019observation, pal2019topological, pu2022topological, apigo2018topological}.

\begin{figure}
    \centering
    \includegraphics[width=0.5\linewidth]{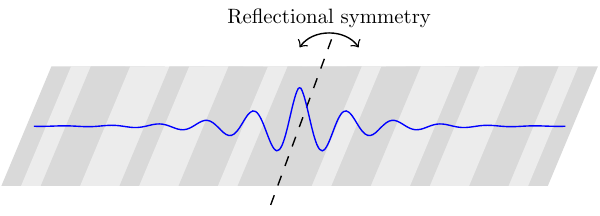}
    \caption{The waveguides studied in this work are formed by taking a pattern of lamina, arranged according to an iterative tiling rule, and introducing an axis of reflectional symmetry. Waves of specific frequencies are strongly guided along the interface formed by the reflection. This sketch shows an example of such a material, along with a representation of a typical decaying wave form. The decaying mode sketched along the axis perpendicular to the axis of reflection; we study one-dimensional problems in this work, considering only this axis.}
    \label{fig:GuideSketch}
\end{figure}

In the setting of localised modes in periodic media, one of the main avenues of investigation is the study of modes that are \emph{topologically protected}. These are localised modes that, due to the topological properties of the underlying periodic medium, experience enhanced robustness with respect to imperfections \cite{khanikaev2013photonic, rechtsman2013photonic}. For non-periodic structures, it is less clear how to define the topological indices which these theories rely on. One quantity of interest is the integrated density of states (IDS), which captures the number of eigenfrequencies below a given frequency and has been shown to be topologically invariant in certain settings. For quasicrystalline systems formed by parametrically modulating a periodic structure (with rational and irrational values of the parameter yielding periodic and quasiperiodic systems, respectively), it has been shown that the IDS can be used to attribute a unique label to each spectral gap \cite{apigo2018topological}. It was, further, shown that this label can be used to predict which gaps will support edge modes in finite crystals \cite{apigo2018topological}. This argument can be understood using the fact that certain parameter values yield commensurate structures and changes in the IDS can be used to predict the number of edge modes that cross the spectral gap in the interval of parameter space between two subsequent commensurate values \cite{xia2020topological}. However, it remains an open question to connect, in general, these topological properties to the robustness of edge modes created in any of the spectral gaps.

Even if the link between topological properties and robustness of edge modes remains unclear for quasicrystalline settings, several studies have highlighted interesting robustness properties. In particular, the robustness of reflection-induced edge modes was explored by \cite{marti2021edge}, who randomly perturbed the positions of a quasicrystalline array of point scatterers on elastic plates. They observed that the reflection-induced edge modes are robust in the sense that they persist in spite of the random imperfections and their eigenfrequencies were ``only slightly shifted''. Some robustness conclusions have also been observed for the edge states in finite-dimensional quasicrystals \cite{liu2022topological, ni2019observation}. We explore robustness of the reflection-induced edge modes that occur in our setting by randomly perturbing the material parameters. Our results suggest that their robustness is comparable to a periodic mode with the same decay rate.

The general results we develop in this work, which are presented in \Cref{sec:general}, apply to any one-dimensional medium that can be constructed by a recursive tiling rule, provided the constituent materials satisfy a symmetry condition (which could be relaxed, in principle, but would lead to less concise expressions). In particular, the material does not necessarily have to be either quasicrystalline or periodic. In Sections~\ref{sec:fibonacci} and~\ref{sec:periodic}, we will present an example of each, for comparison. The quasicrystalline example will be based on a Fibonacci material, building on the previous studies of such structures. We will then present a simple periodic example, in \Cref{sec:periodic}, before studying the relative robustness of the quasicrystalline and periodic materials in \Cref{sec:robust}.

\section{General theory of reflection-induced edge modes} \label{sec:general}

\subsection{Problem setting}

We consider a class of media that are constructed using iterative tiling algorithms that repeatedly join together two different bounded media. We label the two bounded media as $A$ and $B$ and assume that they both have unit width, such that they can be described by the functions
\begin{equation}
    c_A:[0,1)\to(0,\infty) \quad\text{and}\quad c_B:[0,1)\to(0,\infty).
\end{equation}
Subsequently, we can introduce some notation to describe media constructed by repeatedly joining $A$ and $B$ together. We note at this early stage that we use the convention $0\notin\N$.
\begin{definition}[Tiling labels]
Given $N\in\N$ and a set of labels $X_1,X_2,\dots,X_N\in\{A,B\}$, we will use the label $X_1X_2\cdots X_N$ for the tiled material described by the function $c_{X_1X_2\cdots X_N}:[0,N)\to(0,\infty)$ which satisfies
\begin{equation*}
    c_{X_1X_2\cdots X_N}(x)=c_{X_n}(x-n+1), \qquad \text{for }x\in[n-1,n), \, n=1,\dots,N.
\end{equation*}
\end{definition}
For example, we will write AB to mean the medium formed by joining $A$ and $B$, so that $c_{AB}(x)$ is equal to $c_A(x)$ for $0\leq x<1$ and equal to $c_B(x-1)$ for $1\leq x<2$. We can now introduce the notion of a reflective recursive medium, which will be the main object of study in this work:

\begin{definition}[Recursive media]
Given two starting labels $A$ and $B$, we define a \emph{sequence of recursive media} to be a sequence of tiling labels $\{M_N:N\in\N\}$ such that $M_1$ is a tiling of finitely many labels from $\{A,B\}$ and, for $N\geq2$, $M_N$ is a tiling of $M_{N-1}$ with a combination of finitely many labels from $\{M_1,\dots,M_{N-1}\}$. That is, $M_N$ is given by
\begin{equation*}
    M_N=M_{N-1}F(M_1,\dots,M_{N-1}),
\end{equation*}
where $F(M_1,\dots,M_{N-1})$ is a tiling of finitely many elements of $\{M_1,\dots,M_{N-1}\}$.
\end{definition}

\begin{definition}[Reflected recursive media]
Given a sequence of recursive media $\{M_N:N\in\N\}$ we define a \emph{reflected recursive medium} to be the medium obtained by letting $N\to\infty$ and then reflecting in $x=0$. That is, the medium described by the function $\c_{M_\infty}:\R\to(0,\infty)$ which is such that 
\begin{equation*}
    \c_{M_\infty}(x)=\begin{cases}
    c_{M_\infty}(x) & x\geq0, \\ c_{M_\infty}(-x) & x<0.
    \end{cases}
\end{equation*}
\end{definition}

The class of reflective recursive media is broad and contains many different types of media. For example both the periodic media $M_\infty=AAAA\dots$ and $M_\infty=ABABABAB\dots$ fall into this framework (see \Cref{sec:periodic} for more on this). In \Cref{sec:fibonacci}, we will study a version of the famous Fibonacci quasicrystal, which is not periodic. The spectral gaps of similar materials were previously studied by \cite{kohmoto1983localization, kohmoto1987critical, guenneau2008acoustic, gei2010wave, gei2018waves, gei2020phononic, hamilton2021effective}. This example will be referred to as the Fibonacci medium and is given by the sequence of recursive media $\{F_N:N\in\N\}$ defined by
\begin{equation} \label{defn:Fibonacci}
    F_1=A, \quad F_2=AB \quad\text{and} \quad F_N=F_{N-1}F_{N-2}\quad\text{for }N\geq3.
\end{equation}
This sequence reads $F_3=ABA$, $F_4=ABAAB$, $F_5=ABAABABA$, and so on. It is closely related to the Fibonacci sequence. In particular, the number of $A$s in $F_N$ is equal to the $N$th Fibonacci number, while the number of $B$s in $F_N$ is give by the $(N-1)$th Fibonacci number. The reflected material formed by this quasicrystalline tiling rule is sketched in \Cref{fig:Fibonacci_sketch}.

\begin{figure}
    \centering
    \includegraphics[width=\textwidth]{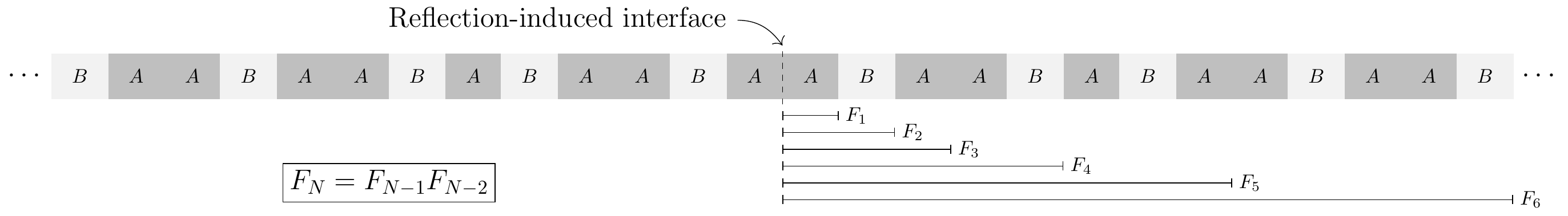}
    \caption{The central quasicrystalline example considered in this work is formed by a Fibonacci-type tiling rule, which is reflected to create an interface. Localised edge modes are created, which decay exponentially away from this interface.}
    \label{fig:Fibonacci_sketch}
\end{figure}

\subsection{Transfer matrices}

We recall some basic properties of transfer matrices, which will be important in what follows. See \cite{SANCHEZSOTO2012191} for a review. Given a one-dimensional Helmholtz equation, posed on a bounded domain of length $L>0$,
\begin{equation} \label{eq:helmholtz_general}
\ddt{u}{x}+ \frac{\omega^2}{c^2(x)}u=0, \quad x\in [0,L],
\end{equation}
where $c\not\equiv0$, we can define a transfer matrix $T(\omega)$ to be such that
\begin{equation} \label{eq:transfer_left}
\begin{pmatrix} u(L) \\ u'(L) \end{pmatrix} =
T(\omega) \begin{pmatrix} u(0) \\ u'(0) \end{pmatrix}.
\end{equation}
We can see that $T$ is given in terms of the functions $\psi_{1,\omega}$ and $\psi_{2,\omega}$:
\begin{align} \label{eq:psi12}
\begin{cases}
\left(\ddt{}{x}+ \frac{\omega^2}{c^2(x)}\right)\psi_{1,\lambda}=0, \\
\psi_{1,\omega}(0) = 1, \\
\psi_{1,\omega}'(0) = 0,
\end{cases}
\quad\text{and}\qquad
\begin{cases}
\left(\ddt{}{x}+ \frac{\omega^2}{c^2(x)}\right)\psi_{2,\lambda}=0, \\
\psi_{2,\omega}(0) = 0, \\
\psi_{2,\omega}'(0) = 1,
\end{cases}
\end{align}
as
\begin{align} \label{eq:Texplicit}
T(\omega)=\begin{pmatrix} \psi_{1,\omega}(L) & \psi_{2,\omega}(L) \\ \psi_{1,\omega}'(L) & \psi_{2,\omega}'(L) \end{pmatrix}.
\end{align}

\begin{lemma} \label{lem:det1}
Any transfer matrix $T$ must be such that, for all frequencies $\omega\in\R$, it holds that $T(\omega)\in\SL$, where $\SL$ is the group of $2\times2$ real matrices with determinant 1.
\end{lemma}
\begin{proof}
Using the definitions of $\psi_{1,\omega}$ and $\psi_{2,\omega}$, we have that 
\begin{align}
    0&\stackrel{\eqref{eq:psi12}}{=} \int_0^L \psi_{1,\omega}(x)\psi_{2,\omega}''(x)-\psi_{2,\omega}(x)\psi_{1,\omega}''(x) \de x\\
    &= \int_0^L \left( \psi_{1,\omega}(x)\psi_{2,\omega}'(x)-\psi_{2,\omega}(x)\psi_{1,\omega}'(x) \right)' \de x \\
    &\stackrel{\eqref{eq:psi12}}{=} \psi_{1,\omega}(L)\psi_{2,\omega}'(L)-\psi_{2,\omega}(L)\psi_{1,\omega}'(L)-1
    \stackrel{\eqref{eq:Texplicit}}{=}\det(T(\omega))-1,
\end{align}
so it follows that $\det(T(\omega))=1$ for all $\omega\in\R$.
\end{proof}

Since the determinant of a matrix is the product of its eigenvalues, elements of $\SL$ fall into one of three categories:
\begin{lemma}[Partition of $\SL$] \label{lem:eigenvalues}
For a given a matrix $T\in\SL$, one of the following must hold:
\begin{enumerate}[noitemsep,topsep=0pt]
    \item $T$ has two real eigenvalues $\lambda_1$ and $\lambda_2$ which satisfy $|\lambda_1|<1$ and $|\lambda_2|>1$. In this case, $|\tr(T)|>2$ and $T$ is said to be hyperbolic.
    \item $T$ has one real eigenvalue $\lambda_1$ which satisfies $\lambda_1=\pm 1$.  In this case, $|\tr(T)|=2$ and $T$ is said to be parabolic.
    \item $T$ has two complex eigenvalues $\lambda_1$ and $\lambda_2$ which satisfy $\lambda_1=\lambda_2^*$ (with $^*$ denoting the complex conjugate). In this case, $|\tr(T)|<2$ and $T$ is said to be elliptic.
\end{enumerate}
\end{lemma}

Several other simple properties will also be useful. Firstly, you can multiply transfer matrices to solve \eqref{eq:helmholtz_general} when two domains have been joined together. That is, if 
\begin{equation*}
\begin{pmatrix} u(L/2) \\ u'(L/2) \end{pmatrix} =
T_1(\omega) \begin{pmatrix} u(0) \\ u'(0) \end{pmatrix}
\quad\text{and}\quad
\begin{pmatrix} u(L) \\ u'(L) \end{pmatrix} =
T_2(\omega) \begin{pmatrix} u(L/2) \\ u'(L/2) \end{pmatrix},
\end{equation*}
then, clearly,
\begin{equation}
\begin{pmatrix} u(L) \\ u'(L) \end{pmatrix} =
T_2(\omega)T_1(\omega) \begin{pmatrix} u(0) \\ u'(0) \end{pmatrix}.
\end{equation}
In addition to this, since $T(\omega)$ is always invertible, we can write that
\begin{equation} \label{eq:transfer_right}
\begin{pmatrix} u(0) \\ u'(0) \end{pmatrix} =
T^{-1}(\omega) \begin{pmatrix} u(L) \\ u'(L) \end{pmatrix}.
\end{equation}
In order to be able to easily relate $T$ and $T^{-1}$, we will assume that $A$ and $B$ are each symmetric. In this case, the straightforward change of variables $x\mapsto-x$ in \eqref{eq:transfer_left} yields the following lemma, which was previously used in \cite{craster2022ssh}:
\begin{lemma}[Symmetric units] \label{lem:symmetry}
Given a function $c:[0,1]\to(0,\infty)$ which is symmetric in the sense that 
\begin{equation} \label{eq:symmetry}
    c(x) = c(1-x), \quad\text{for all } x\in[0,1],
\end{equation}
then the associated transfer matrix, as defined in \eqref{eq:transfer_left}, satisfies
\begin{equation*}
    T(\omega)^{-1}=ST(\omega)S, \quad\text{where } 
    S=\begin{pmatrix} 1 & 0 \\ 0 & -1 \end{pmatrix}.
\end{equation*}
\end{lemma}

These properties mean that for reflected recursive media, as defined above, if we know the transfer matrices $T_A$ and $T_B$ corresponding to the starting media $A$ and $B$, then we are able to describe fully how waves propagate through the medium in either direction. The assumption that $A$ and $B$ are symmetric could be relaxed, however the resulting formulas would be less concise. In the examples presented in Sections~\ref{sec:fibonacci} and~\ref{sec:periodic}, we will take $A$ and $B$ to be homogeneous materials, which is already sufficient to reveal spectra with many rich and exotic properties.

\subsection{Main spectral results}

In general, we will use the notation $T_X$ to denote the transfer matrix associated to the medium with label $X$. In particular, given a sequence of recursive media $\{M_N:N\in\N\}$ we can define a corresponding sequence of transfer matrices $\{T_{M_N}:N\in\N\}$ to be such that 
\begin{equation} \label{defn:TN}
\begin{pmatrix} u(N) \\ u'(N) \end{pmatrix} =
T_{M_N}(\omega) \begin{pmatrix} u(0) \\ u'(0) \end{pmatrix},
\end{equation}
and we have that if $M_N=M_{N-1}X_1^NX_2^N\cdots X_{m(N)}^N$ then $T_{M_N}= T_{X_{m(N)}^N}\cdots T_{X_2^N}T_{X_1^N}T_{M_{N-1}}$.

\begin{lemma}[Symmetric propagation] \label{lem:minusN}
Given a sequence of recursive media $\{M_N:N\in\N\}$ which is based on intial media that are symmetric in the sense of \eqref{eq:symmetry}, it holds that
\begin{equation*}
\begin{pmatrix} u(-N) \\ u'(-N) \end{pmatrix} =
ST_{M_N}(\omega) S\begin{pmatrix} u(0) \\ u'(0) \end{pmatrix}, \quad\text{for all }N\in\N.
\end{equation*}
\end{lemma}
\begin{proof}
Let $N\in\N$ and suppose that we have some sequence of labels $X_1^N,X_2^N,\dots,X_M^N\in\{A,B\}$, with $M=M(N)=\sum_{n=1}^N m(n)$, such that $M_N=X_1^NX_2^N\dots X_M^N$. Then, it holds that $T_{M_N}=T_{X_M^N}\dots T_{X_2^N}T_{X_1^N}$. Conversely, we have that
\begin{align*}
    \begin{pmatrix} u(-N) \\ u'(-N) \end{pmatrix} = T_{X_M^N}^{-1}\dots T_{X_2^N}^{-1}T_{X_1^N}^{-1}
\begin{pmatrix} u(0) \\ u'(0) \end{pmatrix}
    =ST_{X_M^N}\dots T_{X_2^N}T_{X_1^N}S
\begin{pmatrix} u(0) \\ u'(0) \end{pmatrix}
    =ST_{M_N}S
\begin{pmatrix} u(0) \\ u'(0) \end{pmatrix},
\end{align*}
where we have used \Cref{lem:symmetry} to infer that $T_{X_n^N}^{-1}=ST_{X_n^N}S$.
\end{proof}

\begin{theorem} \label{thm:spectrum}
Suppose that the function $\c_{M_\infty}$ describes a reflected recursive medium which is based on initial media that are symmetric in the sense of \eqref{eq:symmetry} and let $(\omega,u)\in\R\times L^2(\R)$ be a solution to the problem
\begin{equation*} 
\ddt{u}{x}+ \frac{\omega^2}{\c_{M_\infty}^2(x)}u=0.
\end{equation*}
Assuming that $\omega\in\R$ is such that
\begin{equation} \tag{I} \label{cond1}
    \min\sigma(T_{M_N}(\omega))\to0, \quad\text{as }N\to\infty,
\end{equation}
then $u(x)\to0$ and $u'(x)\to0$ as $x\to\pm\infty$ with $u\not\equiv 0$ if and only if the normalised eigenvector associated to the smallest eigenvalue of $T_{M_N}(\omega)$, denoted by $(v_{1,N} \ v_{2,N})^\top$, satisfies
\begin{equation} \tag{II} \label{cond2}
    \text{either} \quad v_{1,N}=o\left(\frac{1}{\max\sigma(T_{M_N}(\omega))}\right) \quad  \text{or} \quad v_{2,N}=o\left(\frac{1}{\max\sigma(T_{M_N}(\omega))}\right) \quad\text{as }N\to\infty.
\end{equation}
\end{theorem}

\begin{proof}
In light of \eqref{defn:TN} and \Cref{lem:minusN}, the property that $u(x)\to0$ and $u'(x)\to0$ as $x\to\pm\infty$ is equivalent to $u(N)\to0$ and $u'(N)\to0$ as $N\to\pm\infty$. Notice also that the values of $u(0)$ and $u'(0)$ fully determine $u(x)$, for a given $\omega$. It is also important to realise that \eqref{cond1} implies that $T_{M_N}$ is hyperbolic for large $N$, as it has an eigenvalue which decays to zero. This means, firstly, that $\max\sigma(T_{M_N}(\omega))\to\infty$ as $N\to\infty$, such that \eqref{cond2} is a decay condition on either $v_{1,N}$ or $v_{2,N}$. Secondly, it means that, without loss of generality, we can assume that $T_{M_N}(\omega)$ is hyperbolic for all $N$.

Suppose first that \eqref{cond1} and \eqref{cond2} hold. Thus, using hyperbolicity, we can write that 
\begin{equation} \label{eq:iter}
\begin{pmatrix} u(N) \\ u'(N) \end{pmatrix} =
\begin{pmatrix} v_{1,N} & w_{1,N} \\ v_{2,N} & w_{2,N} \end{pmatrix}
\begin{pmatrix} \lambda_{1,N} & 0 \\ 0 & \lambda_{2,N} \end{pmatrix}
\begin{pmatrix} w_{2,N} & -w_{1,N} \\ -v_{2,N} & v_{1,N} \end{pmatrix}
\begin{pmatrix} u(0) \\ u'(0) \end{pmatrix},
\end{equation}
where $|\lambda_{1,N}|<1$ and $\lambda_{2,N}=\lambda_{1,N}^{-1}$ are the eigenvalues of $T_{M_N}(\omega)$ with associated normalised eigenvectors $(v_{1,N}, v_{2,N})^\top$ and $(w_{1,N}, w_{2,N})^\top$. Some algebra gives
\begin{equation} \label{eq:iter_plus}
\begin{pmatrix} u(N) \\ u'(N) \end{pmatrix} =
\lambda_{1,N}
\begin{pmatrix} v_{1,N} \\ v_{2,N} \end{pmatrix}
\begin{pmatrix} w_{2,N} & -w_{1,N} \end{pmatrix}
\begin{pmatrix} u(0) \\ u'(0) \end{pmatrix}
+
\lambda_{2,N}
\begin{pmatrix} w_{1,N} \\ w_{2,N} \end{pmatrix}
\begin{pmatrix} -v_{2,N} & v_{1,N} \end{pmatrix}
\begin{pmatrix} u(0) \\ u'(0) \end{pmatrix}.
\end{equation}
We can use \Cref{lem:minusN} to derive a similar formula for $-N$:
\begin{equation} \label{eq:iter_neg}
\begin{pmatrix} u(-N) \\ -u'(-N) \end{pmatrix} =
\lambda_{1,N}
\begin{pmatrix} v_{1,N} \\ v_{2,N} \end{pmatrix}
\begin{pmatrix} w_{2,N} & -w_{1,N} \end{pmatrix}
\begin{pmatrix} u(0) \\ -u'(0) \end{pmatrix}
+
\lambda_{2,N}
\begin{pmatrix} w_{1,N} \\ w_{2,N} \end{pmatrix}
\begin{pmatrix} -v_{2,N} & v_{1,N} \end{pmatrix}
\begin{pmatrix} u(0) \\ -u'(0) \end{pmatrix}.
\end{equation}
Assumption \eqref{cond1} means that $\lambda_{1,N}\to0$ as $N\to\infty$ so first term in both \eqref{eq:iter_plus} and \eqref{eq:iter_neg} vanishes for large $N$. We can take care of the second terms thanks to \eqref{cond2}, since choosing either $u(0)=0$ (if $v_{1,N}\to0$) or $u'(0)=0$ (if $v_{2,N}\to0$) means these terms will also vanish as $N\to\infty$. Therefore, $u(x)\to0$ and $u'(x)\to0$ as $x\to\pm\infty$.

Conversely, suppose that \eqref{cond1} holds and that $u(x)\to0$ and $u'(x)\to0$ as $x\to\pm\infty$. Once again, we can use hyperbolicity and use \eqref{eq:iter_plus} and \eqref{eq:iter_neg} to see that, thanks to \eqref{cond1},
\begin{equation} \label{eq:iter_}
\lim_{N\to\infty}
\lambda_{2,N}
\begin{pmatrix} -v_{2,N} & v_{1,N} \end{pmatrix}
\begin{pmatrix} u(0) \\ u'(0) \end{pmatrix}=0
\quad\text{and}\quad
\lim_{N\to\infty}
\lambda_{2,N}
\begin{pmatrix} -v_{2,N} & v_{1,N} \end{pmatrix}
\begin{pmatrix} u(0) \\ -u'(0) \end{pmatrix}=0.
\end{equation}
Adding and subtracting these gives that
\begin{equation}
    \left(\lim_{N\to\infty}\lambda_{2,N} \, v_{2,N}\right) u(0)=0
    \quad\text{and}\quad
    \left(\lim_{N\to\infty}\lambda_{2,N} \, v_{1,N}\right) u'(0)=0.
\end{equation}
Since $(v_{1,N} \ v_{2,N})^\top$ is a unit vector (for all finite $N$) and \eqref{cond1} implies that $\lambda_{2,N}\to\infty$, we must have that \eqref{cond2} holds in order for $u(0)$ and $u'(0)$ to not be both zero.
\end{proof}

\Cref{thm:spectrum} is the main result of this work, the implications of which we will explore in the following sections.  We will refer to the condition \eqref{cond1} as the \textit{spectral gap} condition as it describes which frequencies are able to propagate arbitrarily far through the medium: those that are not are said to be in a spectral gap of the structure. The condition \eqref{cond2} will be referred to as the \textit{edge mode} condition. While \eqref{cond1} describes the required behaviour of the eigenvalues of the sequence of transfer matrices, \eqref{cond2} is the required condition for the solution to be in the eigenspace corresponding to the decaying eigenvalue. The analogue of this result for periodic structures was proved in Theorem~2.4 of \cite{craster2022ssh}.

The main challenge with understanding the conditions \eqref{cond1} and \eqref{cond2} is dealing with the fact that, typically, $\|T_{M_N}\|\to\infty$ as $N\to\infty$. This property is well known in the setting of random media \cite{furstenberg1963noncommuting}. We can employ some tricks to make this easier to handle, both numerically and for making analytic statements. For the edge mode condition \eqref{cond2} we can normalise the transfer matrices $T_{M_N}$ to have unit norm, since the normalised eigenvectors are unchanged. For the spectral gap condition \eqref{cond1} it will be more convenient to study the trace of the transfer matrices, instead of their minimum eigenvalue. Using \Cref{lem:eigenvalues}, and the fact that the trace of a matrix is the sum of its eigenvalues, we have the following result:

\begin{lemma}[Transfer matrix traces] \label{lem:trace}
Given a sequence of transfer matrices $\{T_N\in\R^{2\times2}:N\in\N\}$, it holds that $\min\sigma(T_{N})\to0$ as $N\to\infty$ if and only if $|\tr(T_N)|\to\infty$ as $N\to\infty$.
\end{lemma}

This is useful as it means the spectral gap condition \eqref{cond1} can be understood by dealing with a scalar quantity which, in the case of the Fibonacci medium studied in the next section, will satisfy a simple recursion relation.

\begin{remark}[Defects]
The convenient form of the edge mode condition \eqref{cond2} is due to the choice of reflectional symmetry for the interface. If other interfaces or defects were chosen, this would need to be modified accordingly. For example, the approach developed here could be combined with the methods of \cite{craster2022ssh} to model localised modes in recursive media with a defect (with finite size) introduced. This defect could, for example, be a patch of local periodicity within the otherwise quasicrystalline structure; a material of this type was proposed by \cite{zhou2019topological} (shown in Fig.~3). Introducing dislocations to periodic structures is a well-known approach for creating edge modes that are topologically protected and which have tunable properties \cite{ammari2020robust, drouot2020defect}.
\end{remark}

\begin{remark}[Bound states in the continuum]
This work does not say anything about the existence of bound states in the continuum. By this we mean states which are localised (so $u(x)\to0$ and $u'(x)\to0$ as $x\to\pm\infty$) but for which neither of the eigenvalues of $T_{M_N}(\omega)$ converges to zero as $N\to\infty$ (so there exists $c>0$ such that $\min\sigma(T_{M_N})\geq c$ for all $N$). These have been studied in a range of different periodic and finite-dimensional systems, where they typically arise due to interactions between two nearby resonances with different quality factors \cite{hsu2016bound, ammari2021boundstates, SchnitzerPorter}.
\end{remark}

\section{Application to a Fibonacci quasicrystal} \label{sec:fibonacci}

\subsection{Definition and basic properties}

The main example that we will consider in this work is the Fibonacci medium that was defined in \eqref{defn:Fibonacci} and is given by the sequence of recursive media $\{F_N:N\in\N\}$ defined by $F_1=A$,  $F_2=AB$ and $F_N=F_{N-1}F_{N-2}$ for $N\geq3$. This sequence is closely related to the Fibonacci sequence. In particular, the number of $A$s in $F_N$ is equal to the $N$th Fibonacci number, the number of $B$s in $F_N$ is given by the $(N-1)$th Fibonacci number and the total number of labels in $F_N$ is given by the $(N+1)$th Fibonacci number. 

For simplicity, we will consider the specific case where $A$ and $B$ are both homogeneous materials, such that $c(x)=1$ in $A$ and $c(x)=r$ in $B$, with $r$ assumed to be positive and not equal to 1. These materials satisfy the required assumptions of symmetry and their associated transfer matrices have convenient expressions, given by
\begin{equation} \label{eq:TransferMatrices}
    T_A=\begin{pmatrix} \cos(\omega) & \frac{1}{\omega} \sin(\omega) \\ -\omega \sin(\omega) & \cos(\omega) \end{pmatrix}
    \qquad\text{and}\qquad
    T_B=\begin{pmatrix} \cos(r\omega) & \frac{1}{r\omega} \sin(r\omega) \\ -r\omega \sin(r\omega) & \cos(r\omega) \end{pmatrix}.
\end{equation}

\subsection{Spectral gap condition}

The spectral properties of the Fibonacci material $F_N$, for $N\in\N$, have been studied \emph{e.g.} by \cite{kohmoto1983localization, kohmoto1987critical, hassouani2006surface, guenneau2008acoustic, gei2010wave, gei2018waves, gei2020phononic}. These studies typically concern the case where some finite $N$ is fixed and then $F_N$ is used as the repeating unit cell in a periodic material. They show several important properties, such as its self similarity \cite{gei2018waves}, the Cantor-like emergence of nested spectral gaps \cite{gei2018waves} and the positions of the largest spectral gaps \cite{gei2010wave, hassouani2006surface}. A similar result was obtained by \cite{hamilton2021effective}, where the positions of the main spectral gaps was predicted using a simple homogenisation approach that counts the number of $A$s and $B$s in each $F_N$.

Given \Cref{lem:trace}, we can recast the spectral gap condition \eqref{cond1} from \Cref{thm:spectrum} in terms of the trace of the transfer matrices. This is helpful since we can derive a recursion relation for the trace of $T_{F_N}$, a result first shown by \cite{kohmoto1983localization} and repeated here for completeness:

\begin{lemma}[Recursion relation] \label{lem:recursion}
Let $x_N=\tr(T_{F_N})$ where $F_N$ is the Fibonacci medium defined in \eqref{defn:Fibonacci}. Then, it holds that
\begin{equation*}
    x_{N+1} = x_Nx_{N-1} - x_{N-2}, \quad\text{for } N\geq3.
\end{equation*}
\end{lemma}
\begin{proof}
We have that $T_{F_N}=T_{F_{N-2}}T_{F_{N-1}}$ for $N\geq3$, which can be rearranged to give $T_{F_{N-2}}^{-1}=T_{F_{N-1}}T_{F_N}^{-1}$. Similarly, $T_{F_{N+1}}=T_{F_{N-1}}T_{F_{N}}$. Adding these two equations gives
\begin{equation}
    T_{F_{N+1}}+T_{F_{N-2}}^{-1}=T_{F_{N-1}}T_{F_{N}}+T_{F_{N-1}}T_{F_N}^{-1}.
\end{equation}
Taking the trace and using the fact that $\tr(T_{F_{N-1}}T_{F_{N}})+\tr(T_{F_{N-1}}T_{F_N}^{-1}) =\tr(T_{F_{N-1}})\tr(T_{F_{N}})$ gives the result. 
\end{proof}

\Cref{lem:trace} tells us that a frequency satisfies the spectral gap condition \eqref{cond1} if the sequence of traces $x_N$ diverges. Two examples of these sequences $x_N$ are shown in \Cref{fig:examplesequence}, one of which remains bounded (its absolute value is always less than 2, at least for the first 100 terms computed) and one which diverges quickly. Understanding the difference between these two cases is challenging. The system appears to exhibit chaotic behaviour, in the sense that a small change in the starting values of the sequence can cause the system to switch between the two cases. This behaviour is often observed in non-linear recursion relations. 

Studying \Cref{fig:examplesequence} might lead to the hypothesis that if $|x_n|>2$ for some $n$, then $x_N$ will diverge as $N\to\infty$. This is compounded by noticing that many sequences behave like the one shown in \Cref{fig:examplesequence2}(a), which is bounded in absolute value by 2 but then diverges quickly as soon as it records a value with magnitude greater than 2. Furthermore, similar results hold for other second-order non-linear recursion relations. For example, it is known that if the recursion relation that defines the Mandelbrot set ever exceeds $\pm2$, then it will diverge. In this case, however, no such result is possible and it turns out that there are periodic orbits which take arbitrarily large values. An example of such an orbit is shown in \Cref{fig:examplesequence2}(b), which is periodic with period 6 and can be chosen such that the multiple of 3 terms take arbitrarily large values. This periodic orbit is used to prove the following lemma, which makes this notion precise.

\begin{figure}
    \centering
    \begin{subfigure}{0.45\linewidth}
    \includegraphics[width=\linewidth]{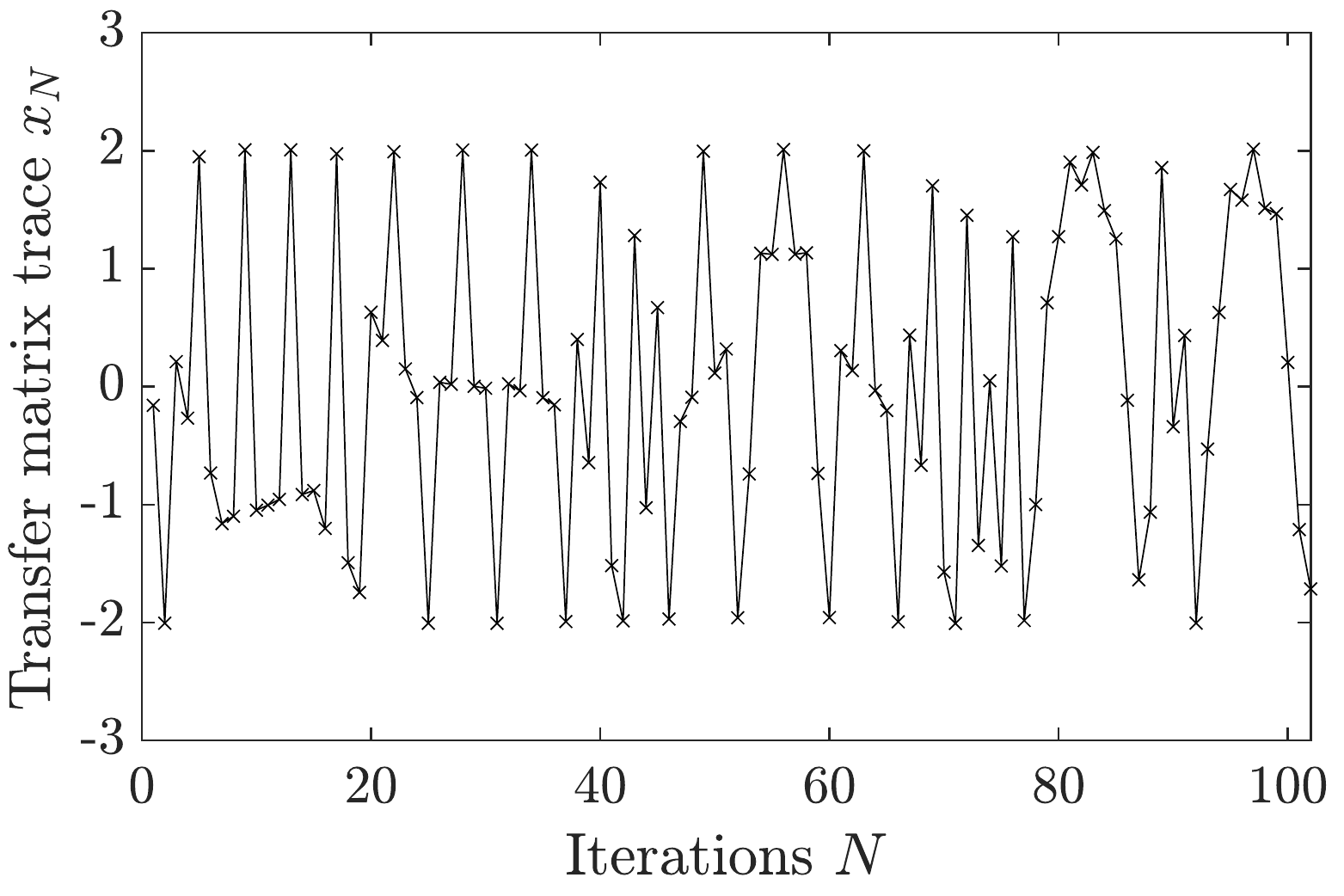}
    \caption{$r=0.92$, $\omega=1.65$.}
    \end{subfigure}
    \hspace{0.1cm}
    \begin{subfigure}{0.45\linewidth}
    \includegraphics[width=\linewidth]{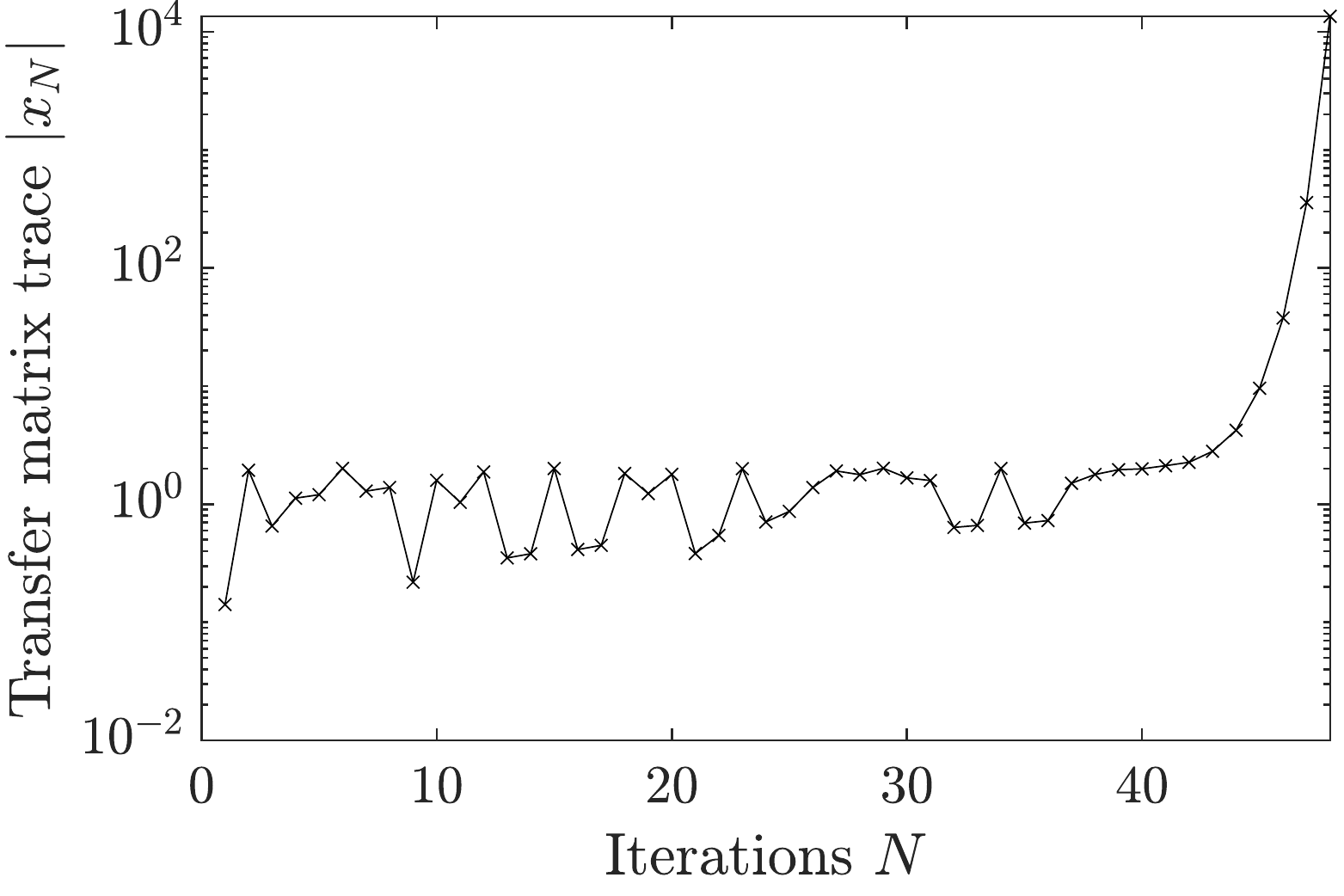}
    \caption{$r=0.92$, $\omega=1.5$.}
    \end{subfigure}
    \caption{The boundedness or divergence of the sequence of transfer matrix traces determines if a frequency is in a spectral gap. For the quasicrystalline Fibonacci medium, (a) shows a frequency for which the sequence is bounded, at least for the first 100 iterations, whereas (b) shows a frequency for which the sequence diverges quickly.}
    \label{fig:examplesequence}
\end{figure}

\begin{lemma}[Arbitrarily large bounded sequences] \label{lem:unboundedsequences}
There exist parameter values $(\omega,r)\in(0,\infty)^2$ such that the solutions to the recursion relation from \Cref{lem:recursion} can be arbitrarily large but the spectral gap condition \eqref{cond1} is not satisfied. That is, given any $K\in(0,\infty)$, we can find $(\omega,r)\in(0,\infty)^2$ such that the corresponding sequence of traces, as defined in \Cref{lem:recursion}, satisfies $K<\max_{N\in\N}|x_N|<\infty$.
\end{lemma}
\begin{proof}
First, notice that if $x_1=0$, $x_2\neq0$ and $x_3=0$, then the recursion relation will give a sequence that is periodic with period 6, given by $0,x_2,0,0,-x_2,0,0,x_2,0,0,-x_2,0,\dots$. If we choose $\omega$ and $r$ such that
\begin{equation}
    \omega=(2n+1)\frac{\pi}{2}\quad\text{and}\quad r=\frac{2m+1}{2n+1}, \quad \text{for } n,m\in\N,
\end{equation}
then we have that $x_1=0$, $|x_2|=r+\tfrac{1}{r}$ and $x_3=0$. Let $K\in(0,\infty)$. For any given $n$, we can choose $m$ sufficiently large such that $|x_2|>K$.
\end{proof}

\begin{figure}
    \centering
    \begin{subfigure}{0.45\linewidth}
    \includegraphics[width=\linewidth]{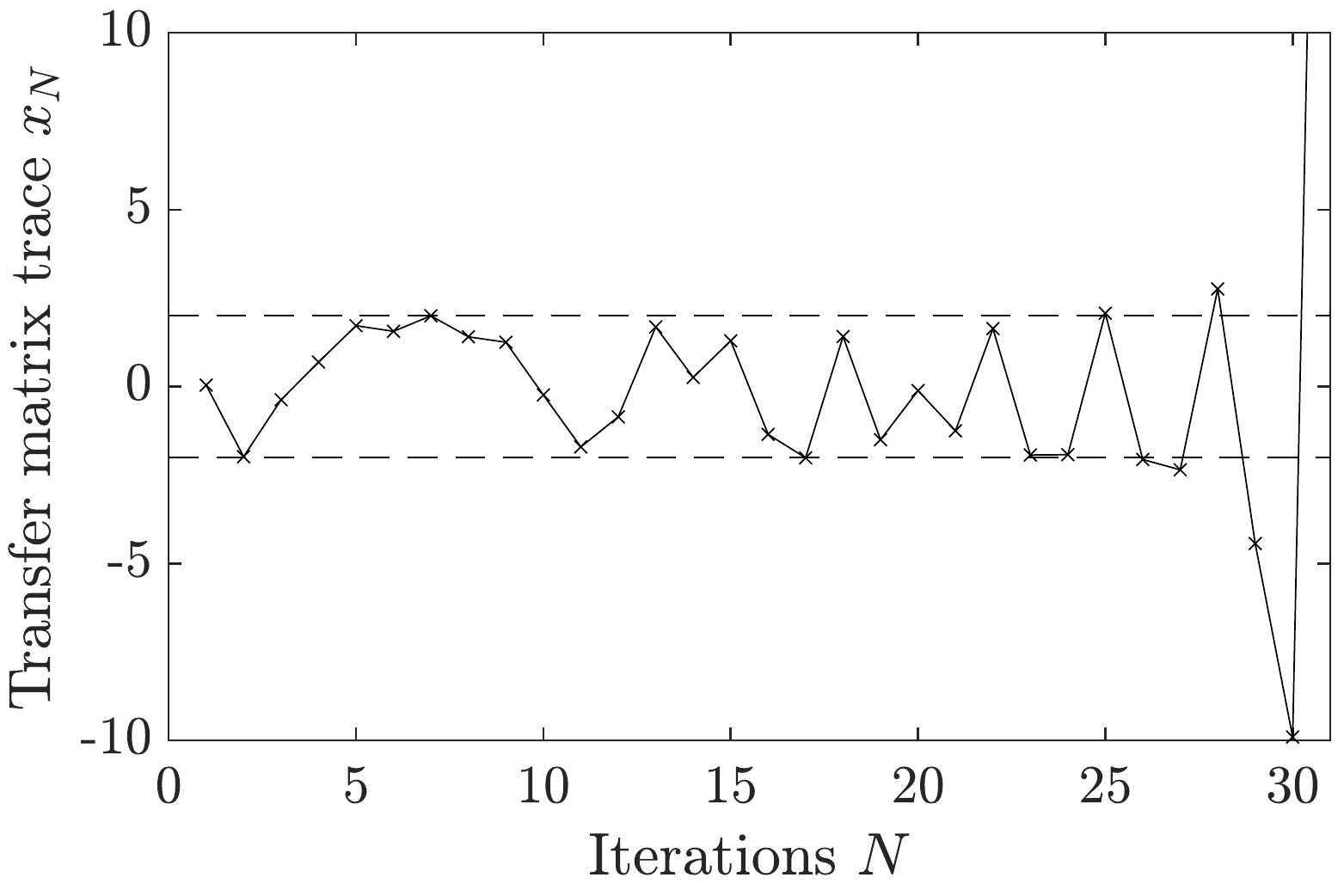}
    \caption{$r=0.92$, $\omega=1.55$.}
    \end{subfigure}
    \hspace{0.1cm}
    \begin{subfigure}{0.45\linewidth}
    \includegraphics[width=\linewidth]{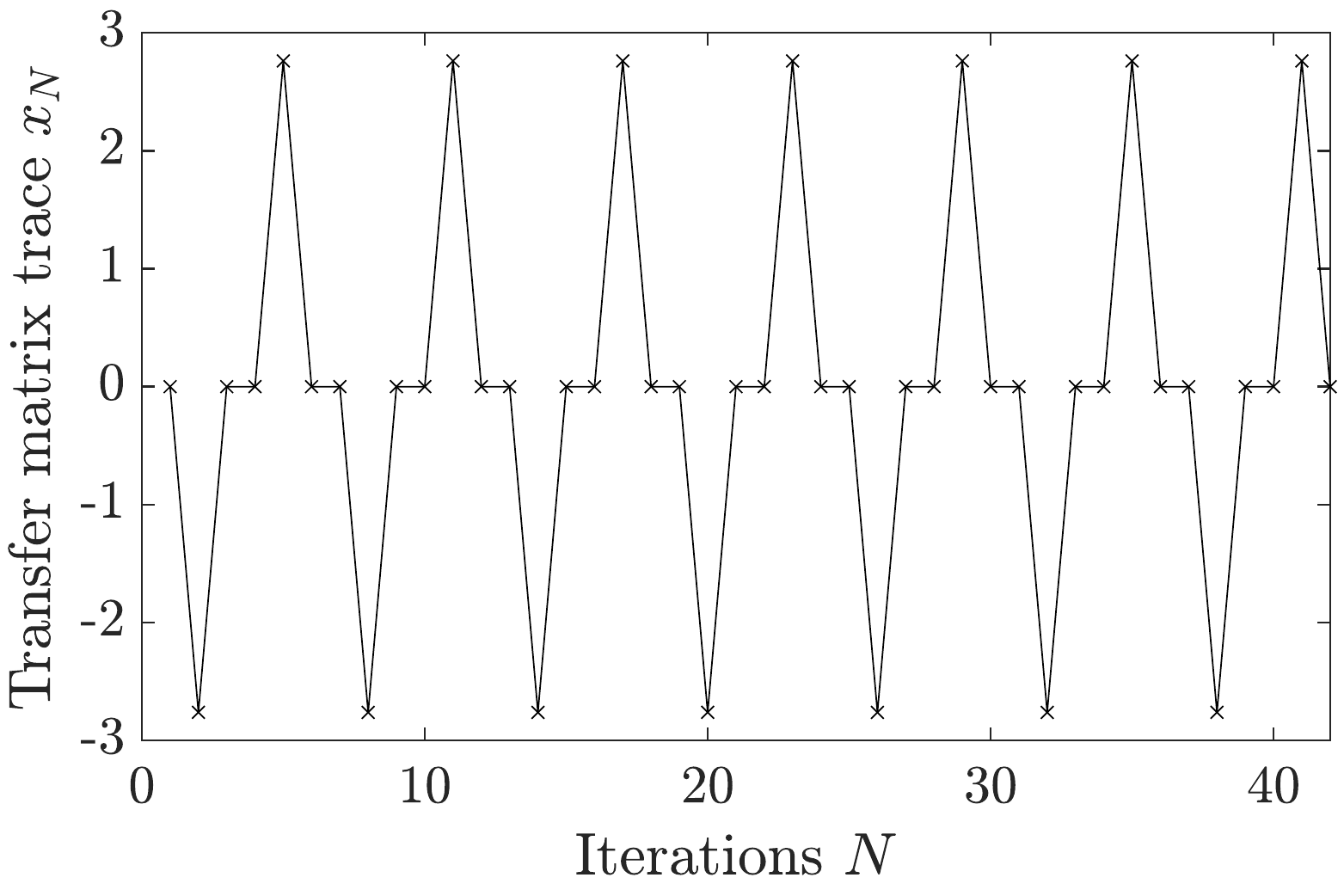}
    \caption{$r=\tfrac{7}{3}$, $\omega=\tfrac{3\pi}{2}$.}
    \end{subfigure}
    \caption{Sequences of transfer matrix traces, corresponding to the quasicrystalline Fibonacci medium, can take arbitrarily large values but be bounded. (a) Many orbits are observed to be bounded by $\pm2$ (shown in dashed lines) before diverging quickly once they escape this region. However, it is not generally true that once an orbit leaves $[-2,2]$ it will diverge. For example, the orbit in (b) is periodic with period 6 but can be chosen so that it periodically takes arbitrarily large values. In this case, every third term alternates between $x_3=-\tfrac{58}{21}\approx-2.76$ and $x_6=\tfrac{58}{21}\approx2.76$. However, increasing $r$ can cause these terms to take a value that is arbitrarily large (see \Cref{lem:unboundedsequences}).}
    \label{fig:examplesequence2}
\end{figure}

In spite of the fact that it is not possible to prove a simple bound for the maximum permitted value of $|x_N|$ in a bounded sequence, we are able to prove a similarly useful result. In particular, the following lemma says that if $x_N$ has left the interval $[-2,2]$ and starts doubling, in the sense that three successive terms are each more than double the previous term, then it will continue doing so indefinitely.

\begin{lemma}[Doubling condition] \label{lem:doubling}
The recursion relation from \Cref{lem:recursion} has the property that if there is some $n\in\N$ such that $2<|x_{n-2}|$, $2|x_{n-2}|<|x_{n-1}|$ and $2|x_{n-1}|<|x_{n}|$, then $2|x_{N}|<|x_{N+1}|$ for all $N\geq n$.
\end{lemma}
\begin{proof}
Let $\gamma=|x_{n-2}|$ and let $\alpha=|x_{n}|$. Then, by hypothesis, we have that $\gamma>2$ and $\alpha>4\gamma$. Subsequently, we see that
\begin{equation}
    |x_{n+1}| \geq |x_n||x_{n-1}| - |x_{n-2}| > 2\alpha\gamma -\gamma =2\alpha(\gamma-1)+2\alpha-\gamma>2\alpha,
\end{equation}
where the last inequality follows since $\gamma-1>1$ and $2\alpha-\gamma>0$. This shows that $|x_{n+1}|>2|x_n|$. For $N>n$, the result follows by induction.
\end{proof}

\Cref{lem:doubling} is the key to our numerical search for parameter values satisfying \eqref{cond1} as we know that once this condition is satisfied we can stop iterating as the sequence is guaranteed to diverge. For example, the sequence in \Cref{fig:examplesequence}(b) satisfies the doubling condition from \Cref{lem:doubling} after 44 iterations. The result of performing this search for different values of $\omega$ and $r$ is shown in \Cref{fig:Fibonacci_gaps}. \Cref{fig:Fibonacci_gaps}(a) shows spectral gaps, where \eqref{cond1} holds, in white. This is designed to be analogous to a Bloch band diagram for a periodic structure. \Cref{fig:Fibonacci_gaps}(b) is the same plot but with the spectral gaps coloured according to the number of iterations that was required for the doubling condition from \Cref{lem:doubling} to be attained.

\begin{figure}
    \centering
    \begin{subfigure}{\linewidth}
    \begin{minipage}{0.7\linewidth}
    \includegraphics[width=\linewidth]{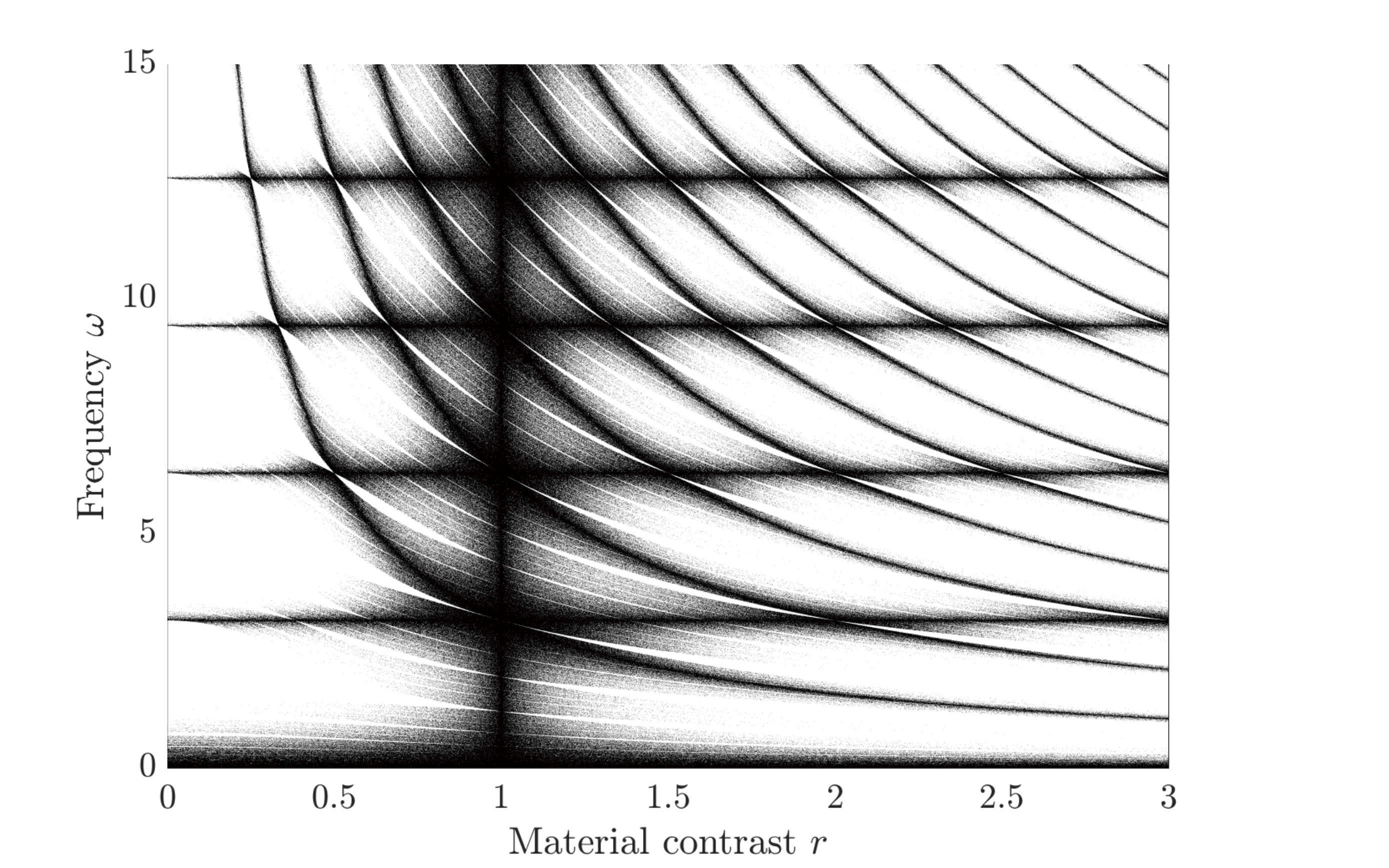}
    \end{minipage}
    \begin{minipage}{0.25\linewidth}
    \includegraphics[width=\linewidth]{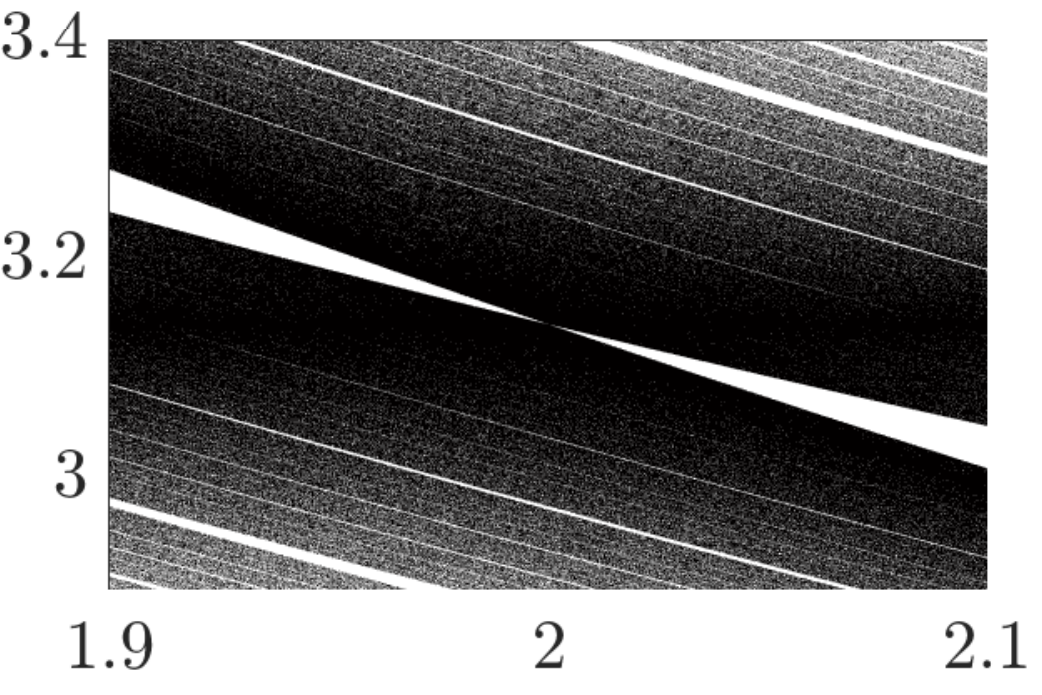}
    
    \vspace{0.4cm}
    
    \includegraphics[width=\linewidth]{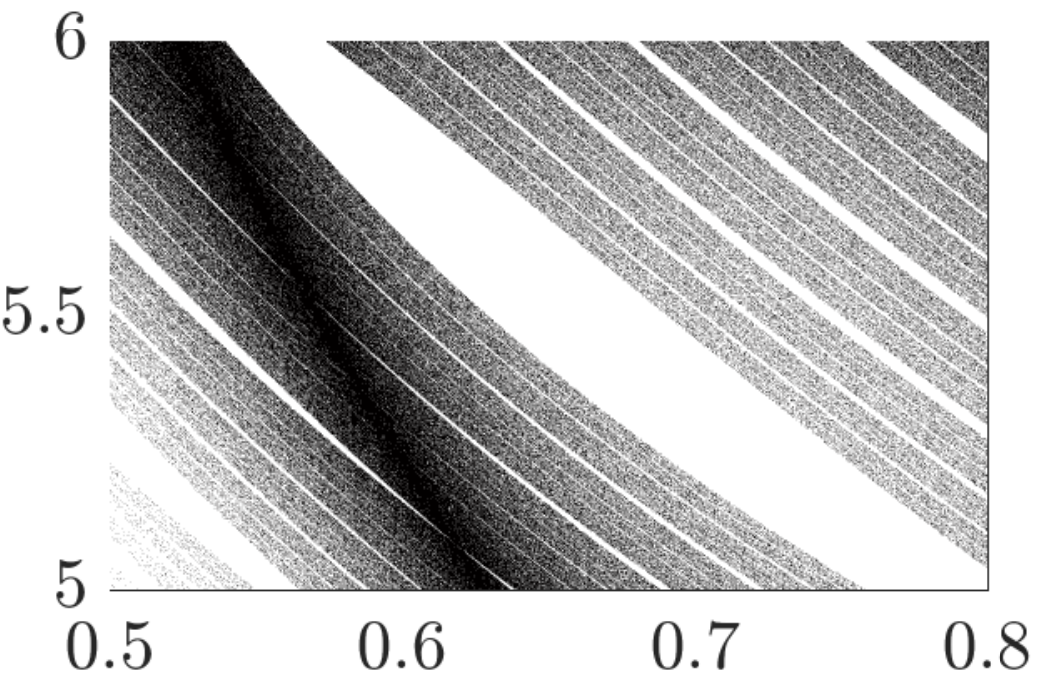}
    \end{minipage}
    \caption{The points at which the spectral gap condition \eqref{cond1} does not hold are shown in black, such that the spectral gaps appear in white. Two zoomed regions are shown in the plots on the right.}
    \end{subfigure}
    
    \vspace{0.3cm}
    
    \begin{subfigure}{\linewidth}
    \begin{minipage}{0.7\linewidth}
    \includegraphics[width=\linewidth]{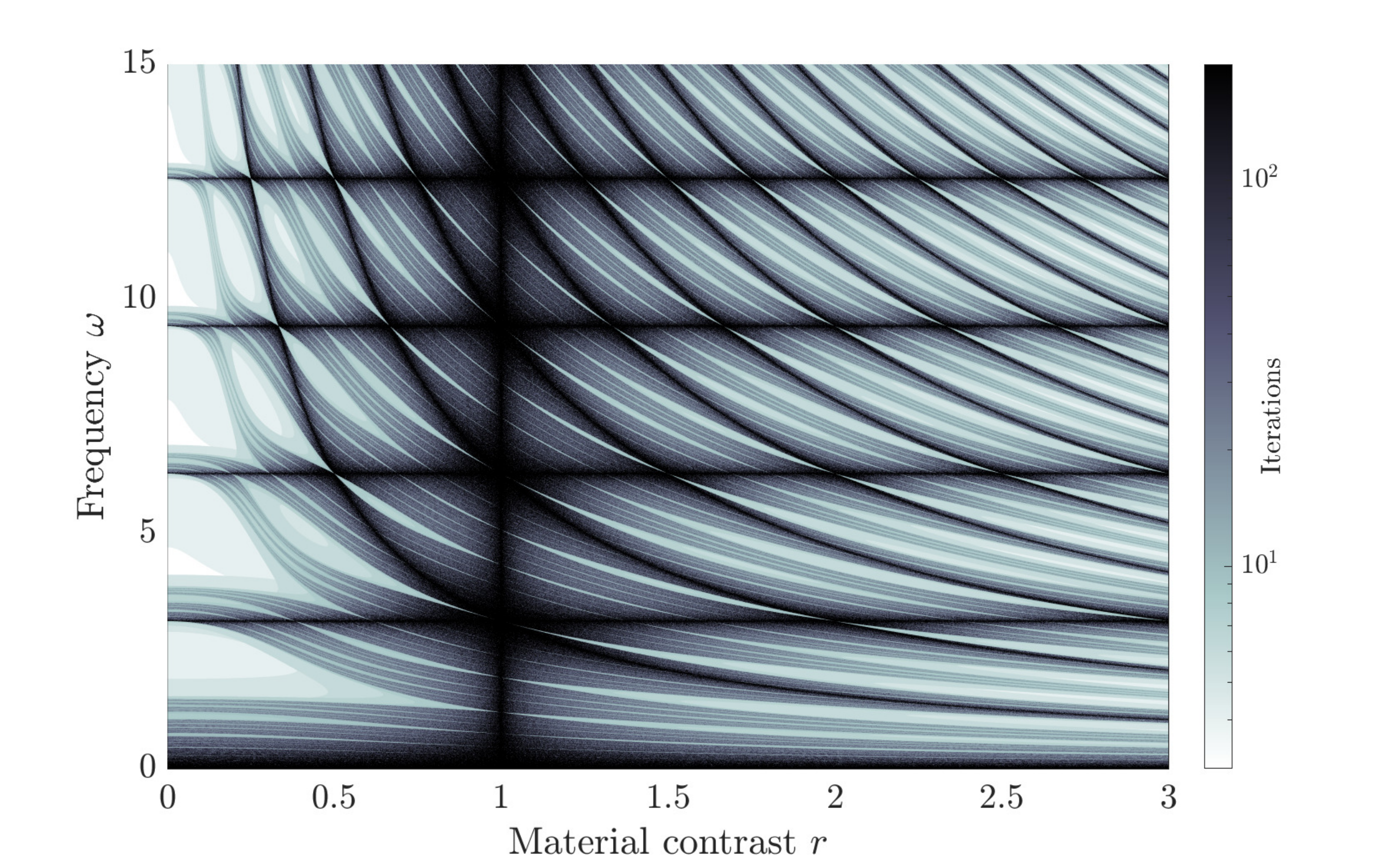}
    \end{minipage}
    \begin{minipage}{0.25\linewidth}
    \includegraphics[width=\linewidth]{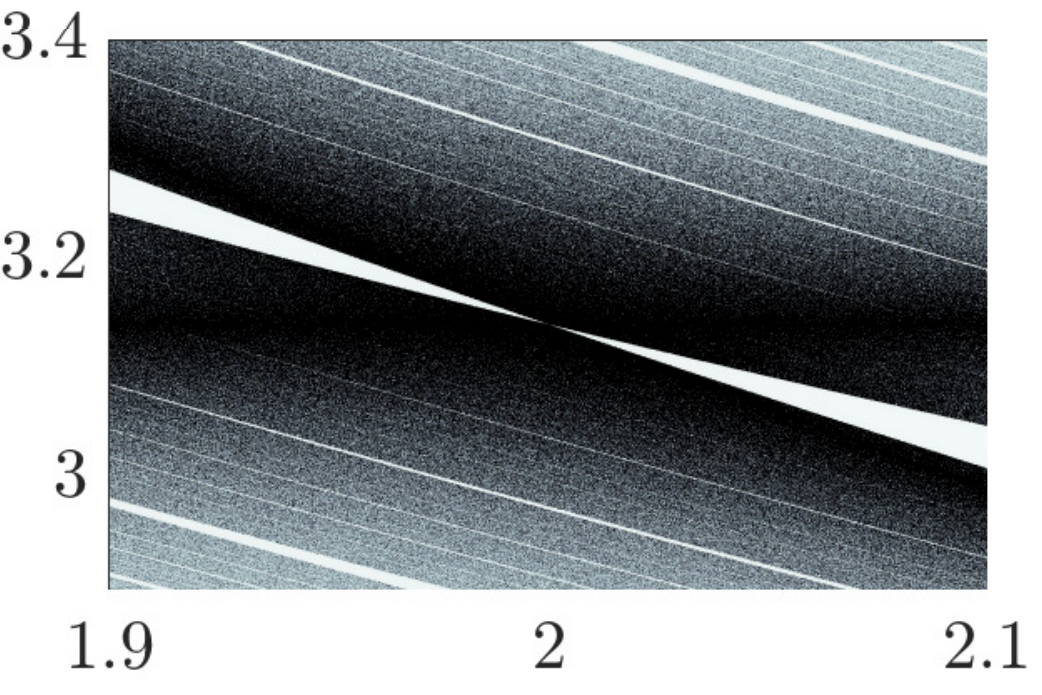}
    
    \vspace{0.4cm}
    
    \includegraphics[width=\linewidth]{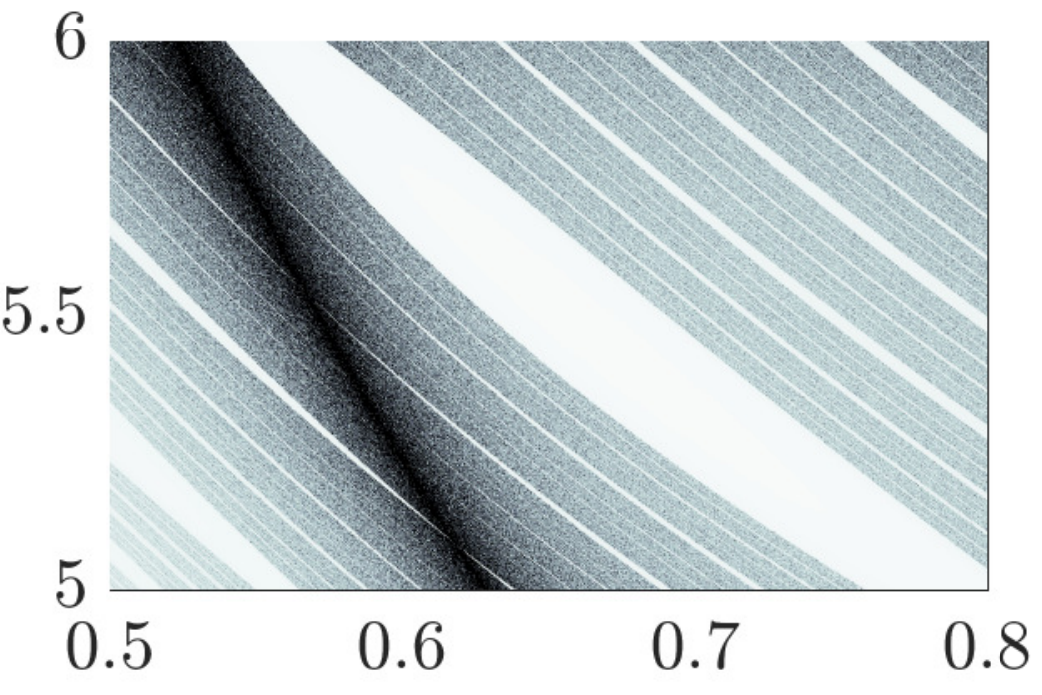}
    \end{minipage}
    
    \caption{The points at which the spectral gap condition \eqref{cond1} does not hold are shown in black and points in the spectral gaps are coloured according to the number of iterations that was required for the doubling condition from \Cref{lem:doubling} to be satisfied. Two zoomed regions are shown in the plots on the right.}
    \end{subfigure}
    \caption{The Fibonacci medium exhibits a complex structure of spectral gaps, that has self-similar and fractal properties. For each pair of material contrast $r$ and frequency $\omega$ in this plot, we check if the spectral gap condition \eqref{cond1} from \Cref{thm:spectrum} holds by iterating the sequence of traces (using \Cref{lem:recursion}) until the doubling condition from \Cref{lem:doubling} is attained.}
    \label{fig:Fibonacci_gaps}
\end{figure}

\Cref{fig:Fibonacci_gaps} shows an exotic pattern of spectral gaps. Particularly from the zoomed in plots (on the right), we can see a fractal-like pattern of gaps emerging. Some of the main features of \Cref{fig:Fibonacci_gaps} can be explained easily. First, the lack of spectral gaps when $r=1$ is due to the two materials, $A$ and $B$, being the same in this case. Additionally, there are horizontal and hyperbolic stripes where there are no spectral gaps. In particular, these are when either $\omega=(2n+1)\tfrac{\pi}{2}$ or $r\omega=(2n+1)\tfrac{\pi}{2}$, for some $n\in\N$, in which case either $T_A$ or $T_B$ is equal to the $2\times2$ identity matrix (up to a factor of $\pm1$). 

\begin{remark}[Rainbow devices]
Given the rich hierarchy of spectral gaps exhibited by the Fibonacci structure, some of which are very large (\emph{cf.}~\cite{zolla1998remarkable}), there are many other potential applications of this material, beyond the symmetry-induced waveguides considered in this work. For example, \Cref{fig:Fibonacci_gaps} shows how the spectral gaps are shifted by modulating the material contrast parameter. This property could be used to create a \emph{rainbow device} which has spatially graded material parameters and, as a result of the local spectral gap being gradually shifted, localises different frequencies at different locations. Rainbow devices based on similar graded materials have been used successfully in many different settings: see \cite{deponti2020graded, davies2021robustness} and references therein.
\end{remark}

\subsection{Edge mode condition}

The edge mode condition \eqref{cond2} is slightly easier to handle than the spectral gap condition \eqref{cond1}. In particular, we can normalise the transfer matrices, since we only need the directions of the eigenvectors, which deals with the issues posed by the fact that $\|T_{M_N}\|\to\infty$ as $N\to\infty$. Thus, given a set of values $\omega$ and $r$ which falls in a spectral gap, we can test if an edge mode is supported by iterating the sequence of transfer matrices and checking if \eqref{cond2} is checking for large iteration numbers. The results of doing this are shown in \Cref{fig:Fibonacci_edges}, where we examine two zoomed in regions of the spectrum and check values in the main spectral gaps for the edge mode condition \eqref{cond2}.

Two examples of the localised edge modes are shown in \Cref{fig:Fibonaccimodes}. These modes are computed for a finite-sized piece of the material, using finite differences. We take material contrast parameter $r=2$ and study a material with 55 pieces of material on either side of the interface (that is, $N=10$ in the Fibonacci tiling algorithm). Second-order finite difference approximations are used, with a step size $h=0.01$, and Dirichlet boundary conditions are imposed at either end of the material. It is noticeable that while the modes fluctuate a little as they decay, particularly at higher frequencies, they nonetheless have a strongly decaying envelope.

\begin{figure}
    \centering
    \begin{tikzpicture}
        \node[inner sep=0pt] at (0,0){\includegraphics[width=.48\textwidth]{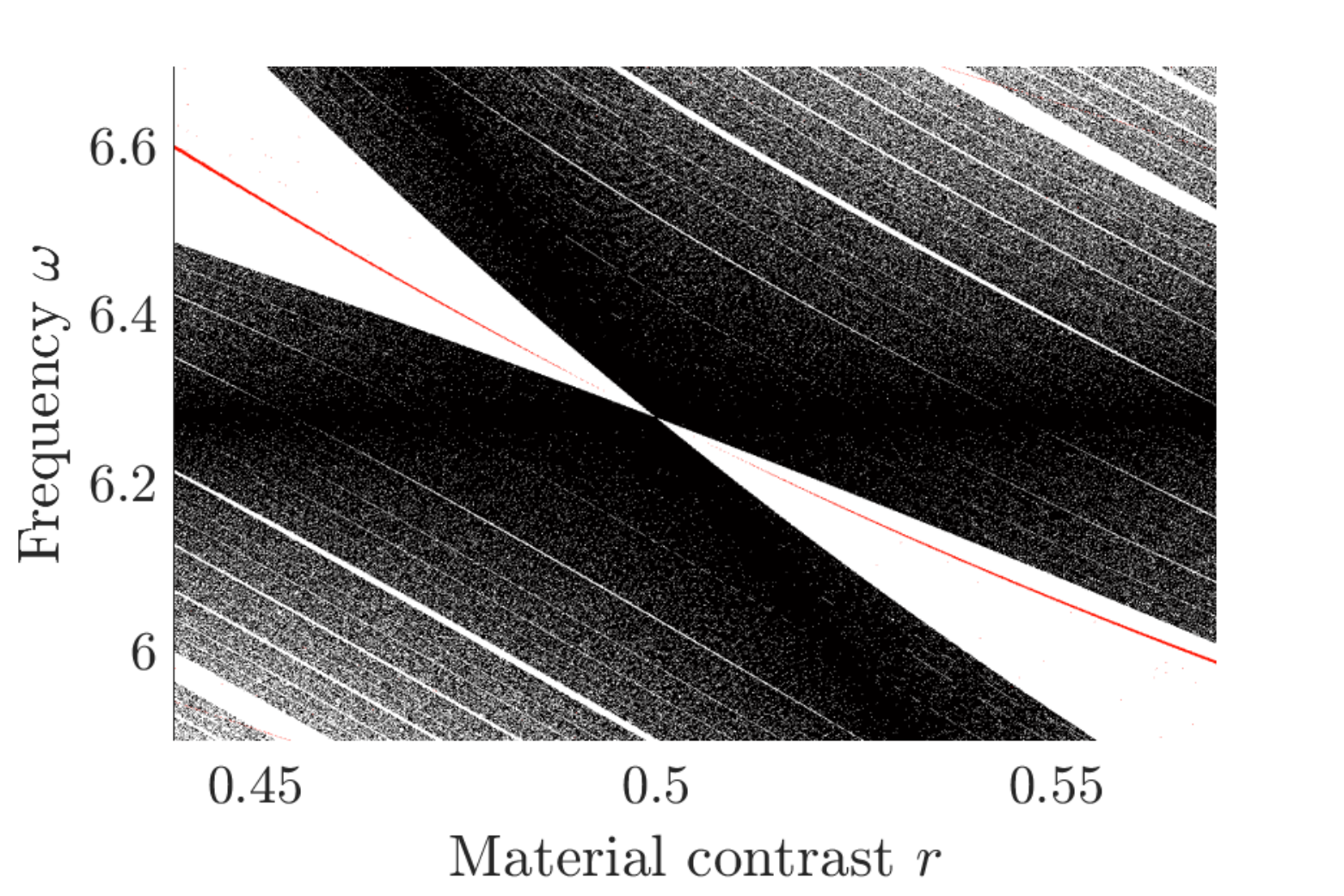}};
        \node[inner sep=0pt] at (8,0){\includegraphics[width=.48\textwidth]{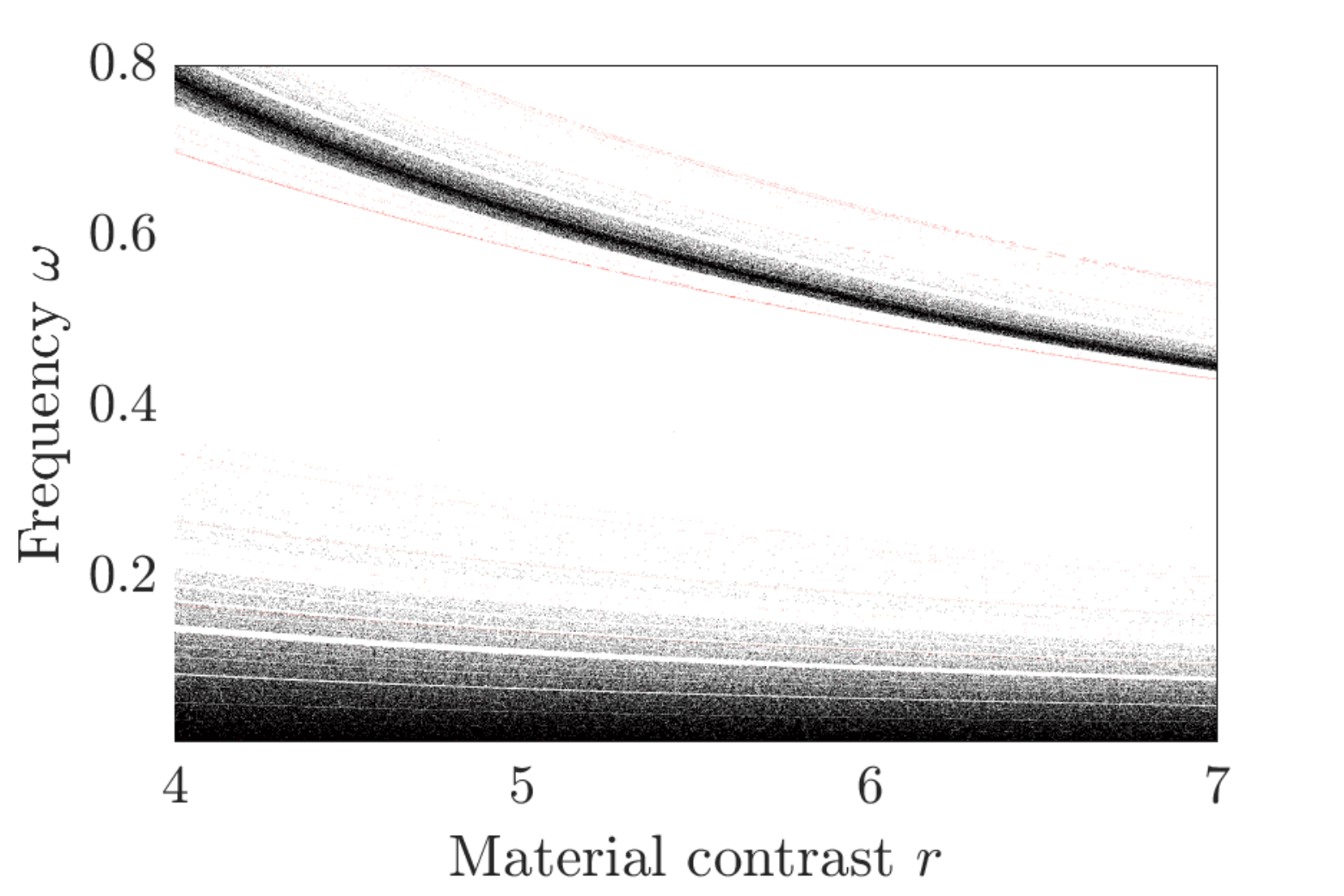}};
        \draw[->, red] (2.8,-1.5) to (2.4,-1);
        \draw[->, red] (-2.2,1.8) to (-2.4,1.5);
        \node at (3.5,-1.7) {\color{red}\small Edge mode};
        \draw[->, red] (7,0.5) to (7.5,1);
        \draw[->, red] (7,0.5) to (7.5,1.7);
        \node at (7,0.3) {\color{red}\small Edge modes};
    \end{tikzpicture}
    \caption{The reflected Fibonacci medium supports localised edge modes within its spectral gaps. We check if the edge mode condition \eqref{cond2} from \Cref{thm:spectrum} holds for different material contrasts $r$ and frequencies $\omega$ lying within spectral gaps, as determined by \eqref{cond1}. Spectral gaps are shown in white, as in \Cref{fig:Fibonacci_gaps}, and the red lines denote the locations of predicted localised edge modes. Two different zoomed in regions of the spectrum are shown.}
    \label{fig:Fibonacci_edges}
\end{figure}
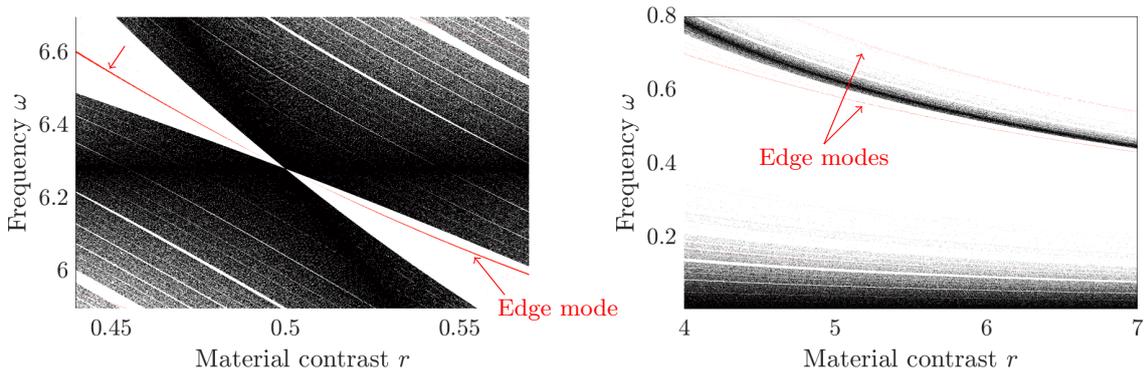

\begin{figure}
    \centering
    \begin{subfigure}{0.45\linewidth}
    \includegraphics[width=\linewidth]{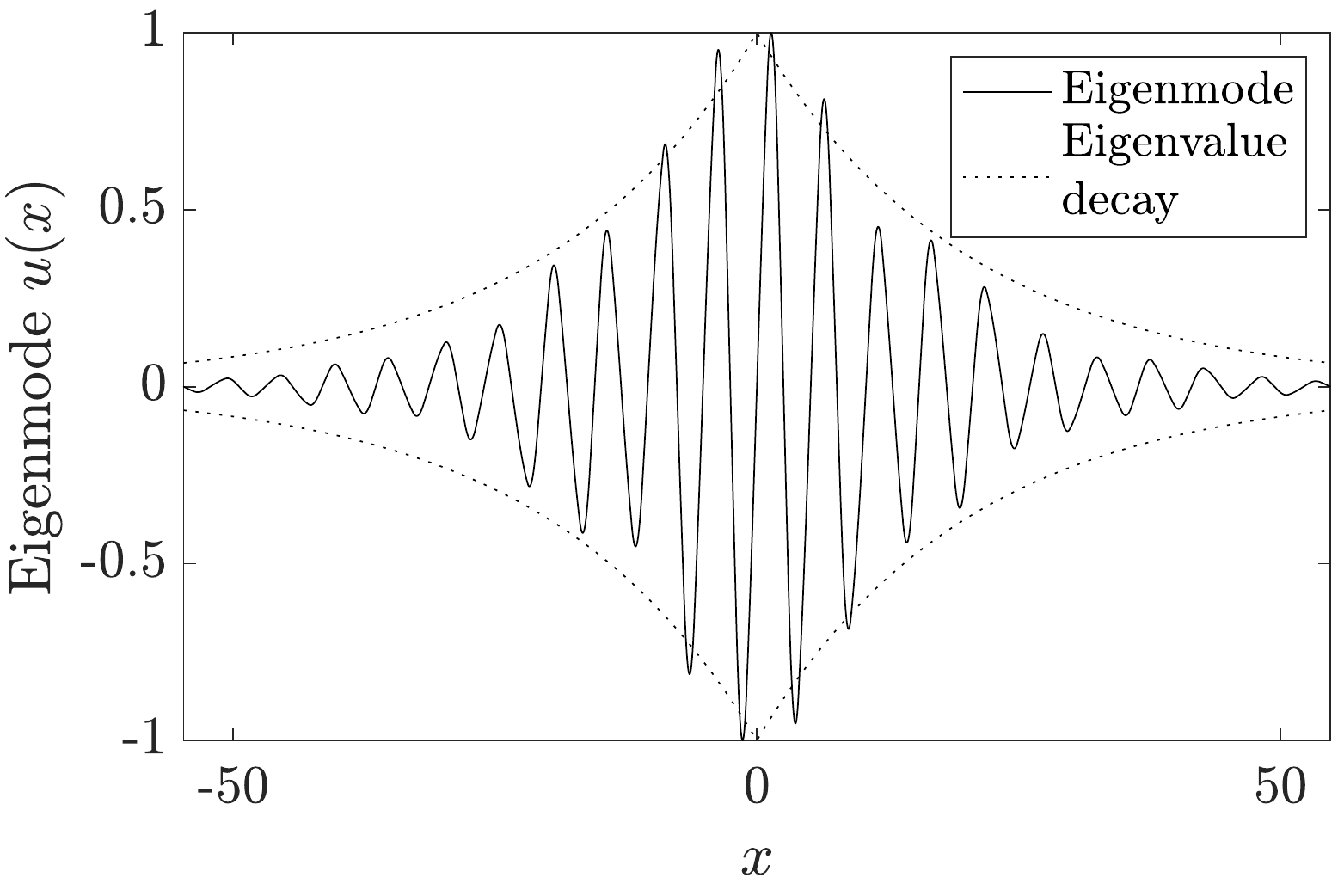}
    \caption{$\omega=1.416$.}
    \end{subfigure}
    \hspace{0.1cm}
    \begin{subfigure}{0.45\linewidth}
    \includegraphics[width=\linewidth]{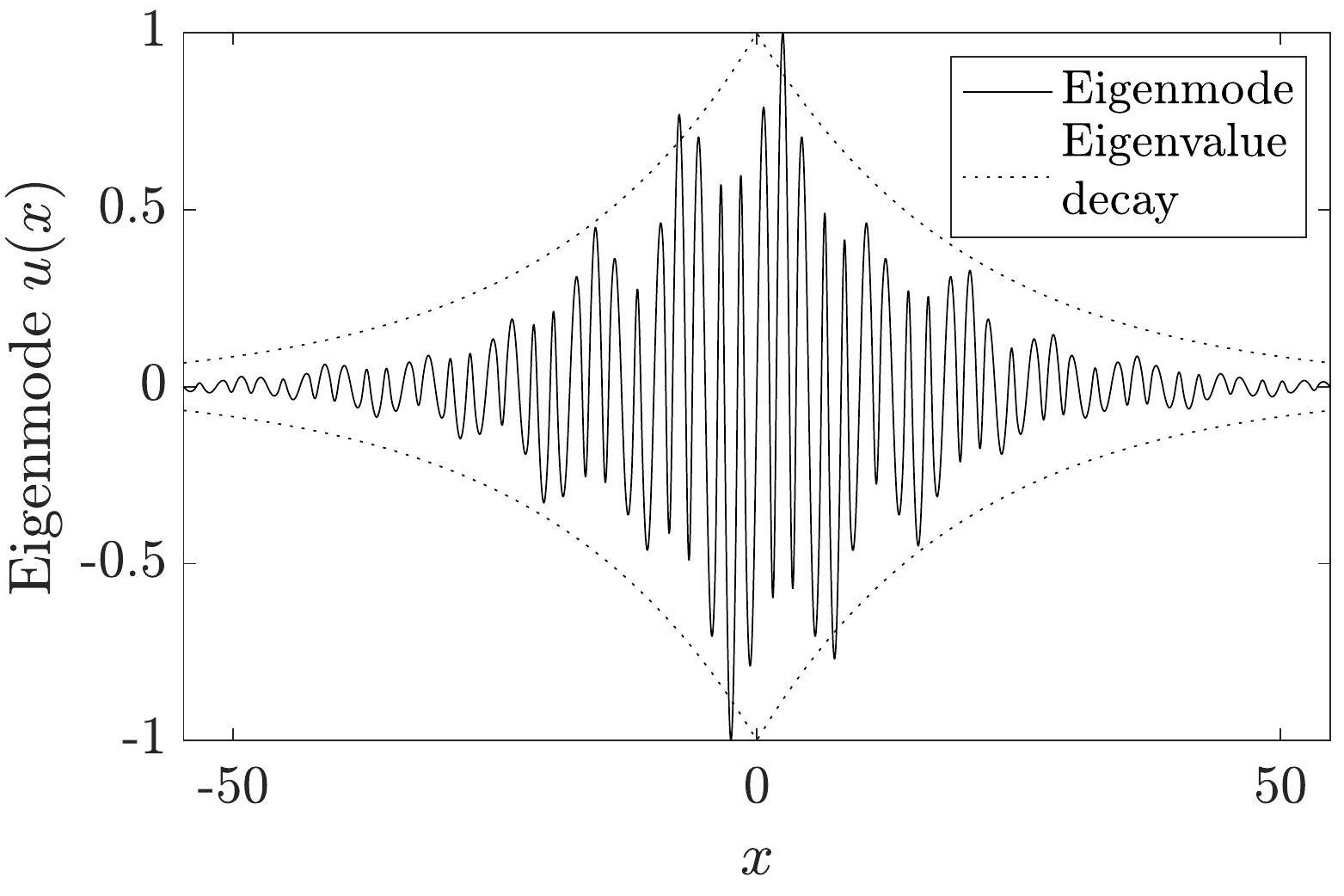}
    \caption{$\omega=4.869$.}
    \end{subfigure}
    \caption{The localised edge modes supported by the reflected Fibonacci medium decay exponentially quickly away from the interface. We take material contrast parameter $r=2$ and study a finite material with 55 pieces of material on either side of the interface (the corresponds to $N=10$). The modes are shown alongside an approximate decay rate (the dotted line), which is approximated using transfer matrix eigenvalues and the formula \eqref{eq:envelope}.}
    \label{fig:Fibonaccimodes}
\end{figure}

In addition to predicting the frequencies at which localised edge modes will exist, it is useful to estimate the rate at which they decay away from the interface. In one-dimensional waveguides based on periodic media, this can be estimated using the eigenvalues of the transfer matrices on either side of the interface, see \cite{craster2022ssh}. In our quasicrystalline setting, an approximate bound on the rate of decay can be similarly achieved by studying the decaying sequence of transfer matrix eigenvalues. In particular, the edge mode condition \eqref{cond2} means the mode lies in the limiting eigenspace of the decaying eigenvalues. Thus, if $N$ is large, we can make the approximation
\begin{equation}
\begin{pmatrix} u(\#_{F_N}) \\ u'(\#_{F_N}) \end{pmatrix} \approx
\min\sigma(T_{F_N}(\omega)) \begin{pmatrix} u(0) \\ u'(0) \end{pmatrix},
\end{equation}
where $\#_{F_N}$ is the number of labels in $F_N$ (and, since $A$ and $B$ are assumed to have unit length, is equal to the length of $F_N$). Assuming that the decay rate is approximately constant across the length of the material, we reach the expression
\begin{equation} \label{eq:envelope}
    \begin{pmatrix} u(x) \\ u'(x) \end{pmatrix} \approx
    \exp\left( \frac{\log(\min\sigma(T_{F_N}))}{\#_{F_N}}|x| \right)
    \begin{pmatrix} u(0) \\ u'(0) \end{pmatrix},
\end{equation}
where we want $N$ to be large, in order to get a reasonable approximation. Note that since $\min\sigma(T_{F_N}))\to0$ as $N\to\infty$, $\log(\min\sigma(T_{F_N})))<0$ for large $N$ so we have an exponentially decaying mode. \Cref{fig:Fibonaccimodes} shows two localised edge modes along with the approximate envelope from \eqref{eq:envelope}. In this case, we take $N=10$ (so $\#_{F_N}=55$) and get a reasonable prediction of the overall shapes of the modes (in spite of the crude approximations that were used to derive \eqref{eq:envelope}).

\begin{remark}[Lyapunov exponents]
In the setting of random media, localisation lengths are typically quantified using Lyapunov exponents \cite{borland1963nature, carmona1982exponential, comtet2013lyapunov, scales1997lyapunov}. Using our notation, if $\#_{M_N}\in\N$ is the number of labels in $M_N$, then the (maximal) Lyapunov exponent is defined as $\gamma = \lim_{N\to\infty} {\log\|T_{M_N}\|}/{\#_{M_N}}$. The fact that $\gamma\geq0$ in general for random media follows from Furstenberg \cite{furstenberg1963noncommuting}. The quantity $1/\gamma$ can then be interpreted as a measure of the localisation length \cite{borland1963nature, carmona1982exponential, furstenberg1963noncommuting}. For the Fibonacci medium studied here, which is deterministic, the Lyapunov exponent can be computed explicitly. It is also easy to see that $0\leq\gamma\leq \varphi^{-1} \log\|T_A\|+\varphi^{-2}\log\|T_B\|$, where $\varphi=(1+\sqrt{5})/2$ is the golden ratio. However, in spite of being easy to work with, the Lyapunov exponent provides a much less tight bound on the decay of the eigenmode than the eigenvalue sequence used in \eqref{eq:envelope}. This is to be expected, since the Lyapunov exponent $\gamma$ doesn't contain any information about the reflection that induced the defect. Conversely, information about the specific choice of defect was crucial to the derivation of \eqref{eq:envelope} (and we could modify the result accordingly for different types of defects and interfaces).
\end{remark}

\section{Application to a periodic material} \label{sec:periodic}

The framework developed in \Cref{sec:general} applies to any material that can be defined by a recursive tiling rule. While quasicrystalline media are the interesting case for which the theory was developed, these methods can also be used to describe periodic media. For example, given an initial material $P_1$ (which we would prefer to be inhomogeneous, to give interesting phenomena), we define $\{P_N:N\in\N\}$ as
\begin{equation} \label{eq:periodic}
    P_{N+1}=P_NP_1.
\end{equation}
The reflected material based on this rule is depicted in \Cref{fig:Periodic_sketch}, where $P_1=AB$.

\begin{figure}
    \centering
    \includegraphics[width=\textwidth]{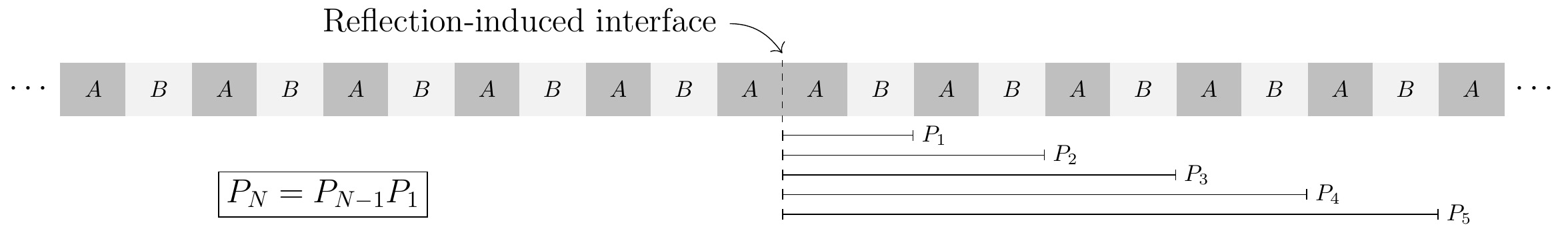}
    \caption{The theory developed in this work can also be used to model reflected periodic media, which support localised edge modes that decay away from the reflection-induced interface.}
    \label{fig:Periodic_sketch}
\end{figure}

In the case of this periodic medium, the conditions \eqref{cond1} and \eqref{cond2} from \Cref{thm:spectrum} are particularly easy to understand. In fact, they reduce to the formulas that were derived in \cite{craster2022ssh}. The periodicity means that successive transfer matrices $T_{P_N}$ are easy to characterise since $T_{P_N}=(T_{P_1})^N$. This means that the spectral gap condition \eqref{cond1}, which says that $\min\sigma(T_{P_N})\to0$, is equivalent to the condition that the minimum eigenvalue of $T_{P_1}$ has magnitude less than 1. This simple condition, leads directly to the pattern of spectral gaps plotted in \Cref{fig:SpectralGapPeriodic}. 

\Cref{fig:SpectralGapPeriodic} also shows, in the two right-hand plots, the localised edge modes that exist for certain parameter values within the spectral gaps. This is based on \eqref{cond2} from \Cref{thm:spectrum}, which is also straightforward in the case of a periodic material since the eigenvectors of $T_{P_N}=(T_{P_1})^N$ are the same as those of $T_{P_1}$. This means we simply need to check if the eigenvector associated to the smaller eigenvalue of $T_{P_1}$ is parallel to either $(1,0)^\top$ or $(0,1)^\top$. Some examples of the periodic modes are shown in \Cref{fig:Periodicmodes}, computed using a finite-difference approximation of a finite-sized piece of the material (analogous to \Cref{fig:Fibonaccimodes}). Here, upper bounds for the decay rates are shown, according to the formula
\begin{equation}
    |u(x)|\leq \exp\left[ \log(\min \sigma(T_{P_1})) |x| \right],
\end{equation}
where $\log(\min \sigma(T_{P_1}))<0$ since $\min \sigma(T_{P_1})<1$. This bound was proved in Corollary~2.5 of \cite{craster2022ssh} (where tighter estimates for the decay rate were computed using high-frequency homogenisation).

\begin{figure}
    \centering
    \begin{minipage}{0.65\linewidth}
    \includegraphics[trim=0 0 2cm 0,clip,width=\linewidth]{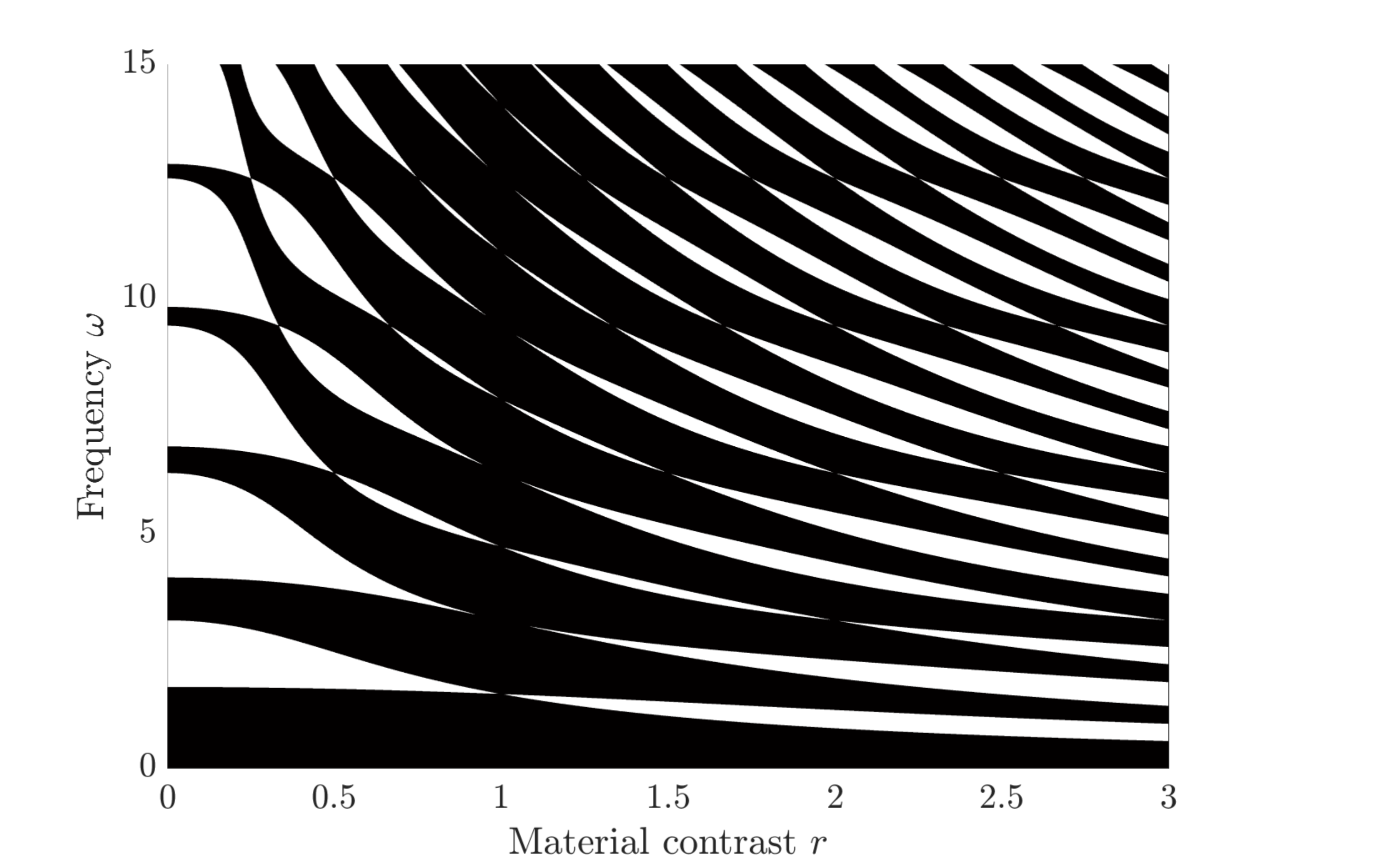}
    \end{minipage}
    \begin{minipage}{0.3\linewidth}
    \includegraphics[width=\linewidth]{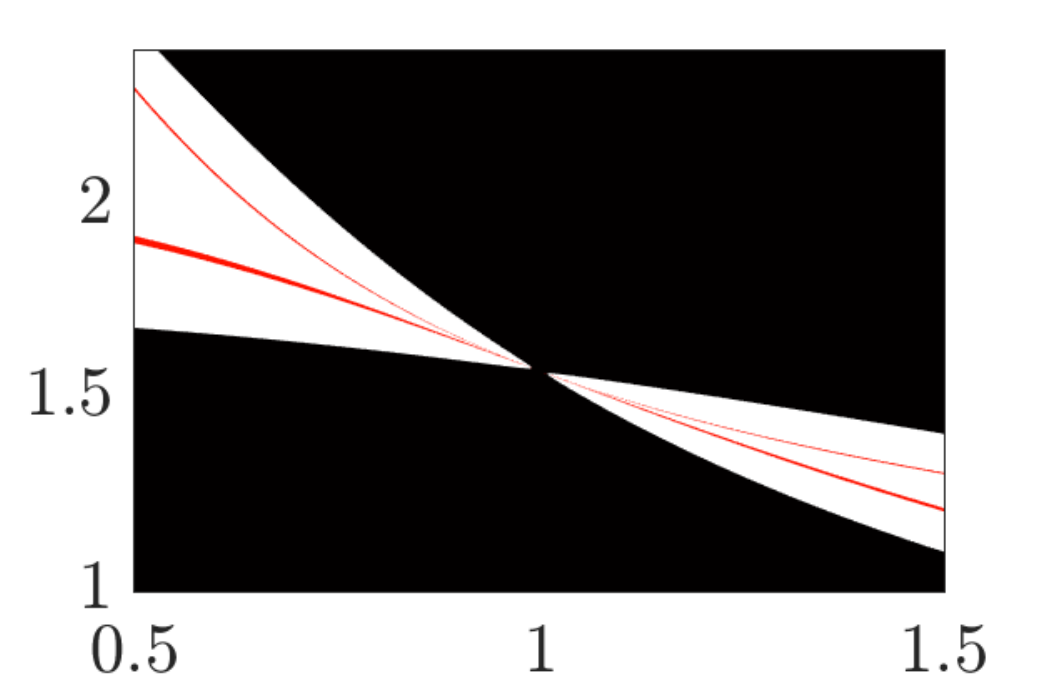}
    
    \vspace{0.2cm}
    
    \includegraphics[width=\linewidth]{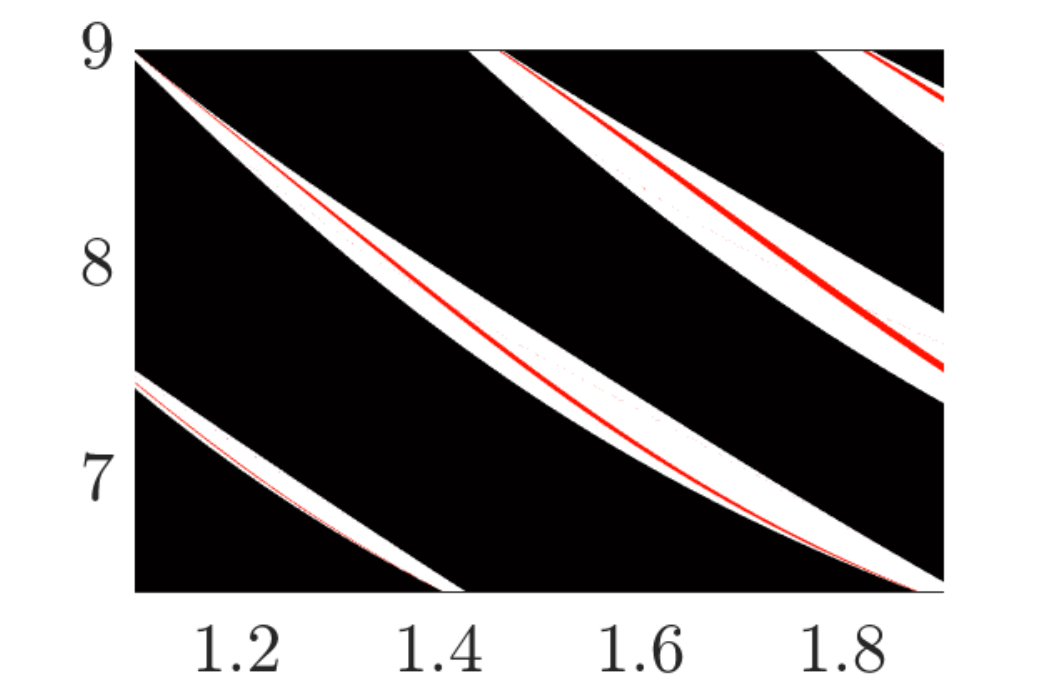}
    \end{minipage}
    \caption{The theory developed in this work can be used to characterise the spectral gaps of periodic media, as well as the localised edge modes created by introducing an axis of reflectional symmetry. On the left, the pattern of spectral gaps is shown, by checking if the condition \eqref{cond1} is satisfied. Spectral gaps are shown in white (this is analogous to \Cref{fig:Fibonacci_gaps}(a) for the Fibonacci material). On the right, values of the material contrast $r$ and the frequency $\omega$, that lie within the spectral gaps, are checked for the edge mode condition \eqref{cond2}. Point supporting edge modes are shown in red (this is the analogue of \Cref{fig:Fibonacci_edges}).}
    \label{fig:SpectralGapPeriodic}
\end{figure}

\begin{figure}
    \centering
    \begin{subfigure}{0.45\linewidth}
    \includegraphics[width=\linewidth]{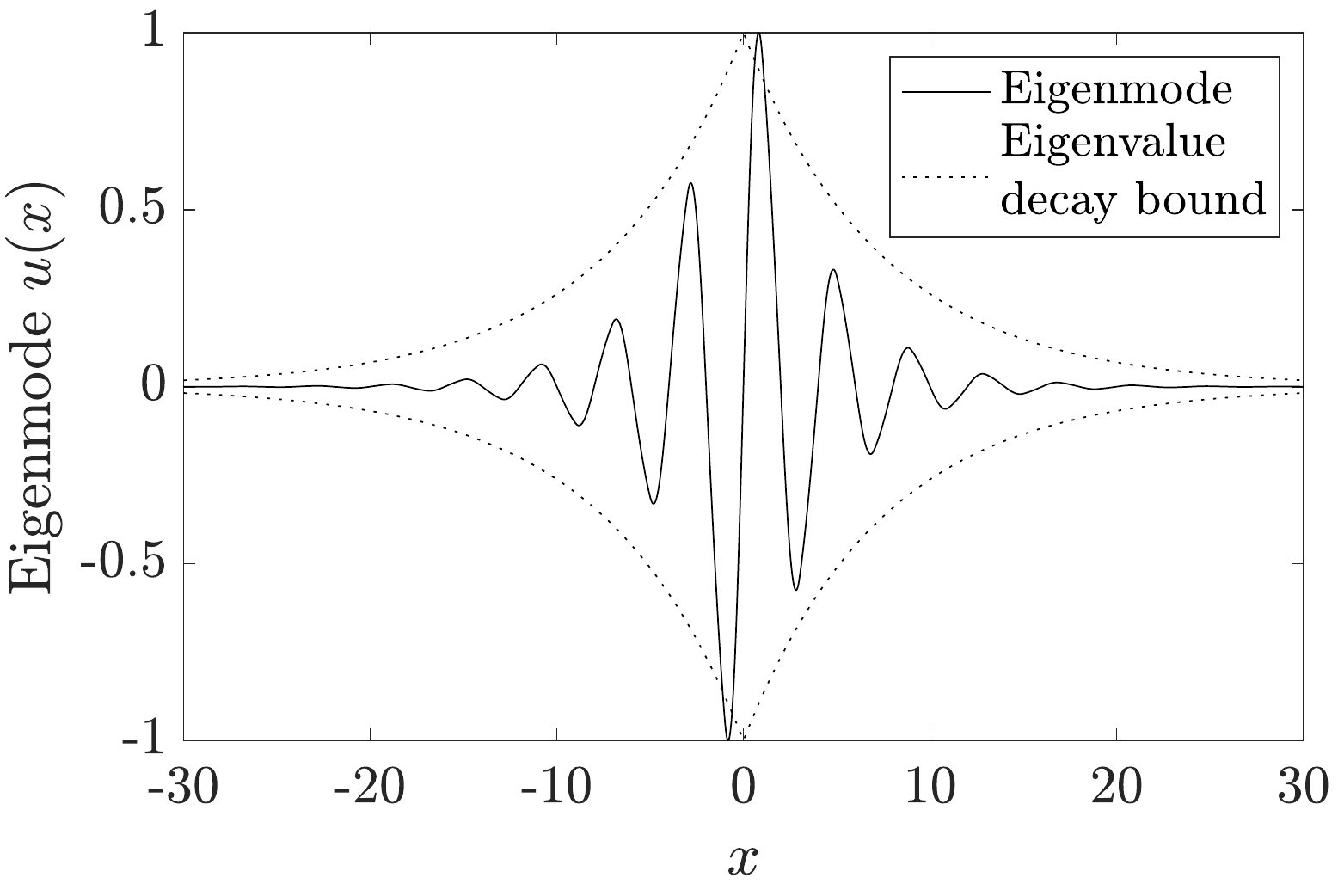}
    \caption{$\omega=1.916$.}
    \end{subfigure}
    \hspace{0.1cm}
    \begin{subfigure}{0.45\linewidth}
    \includegraphics[width=\linewidth]{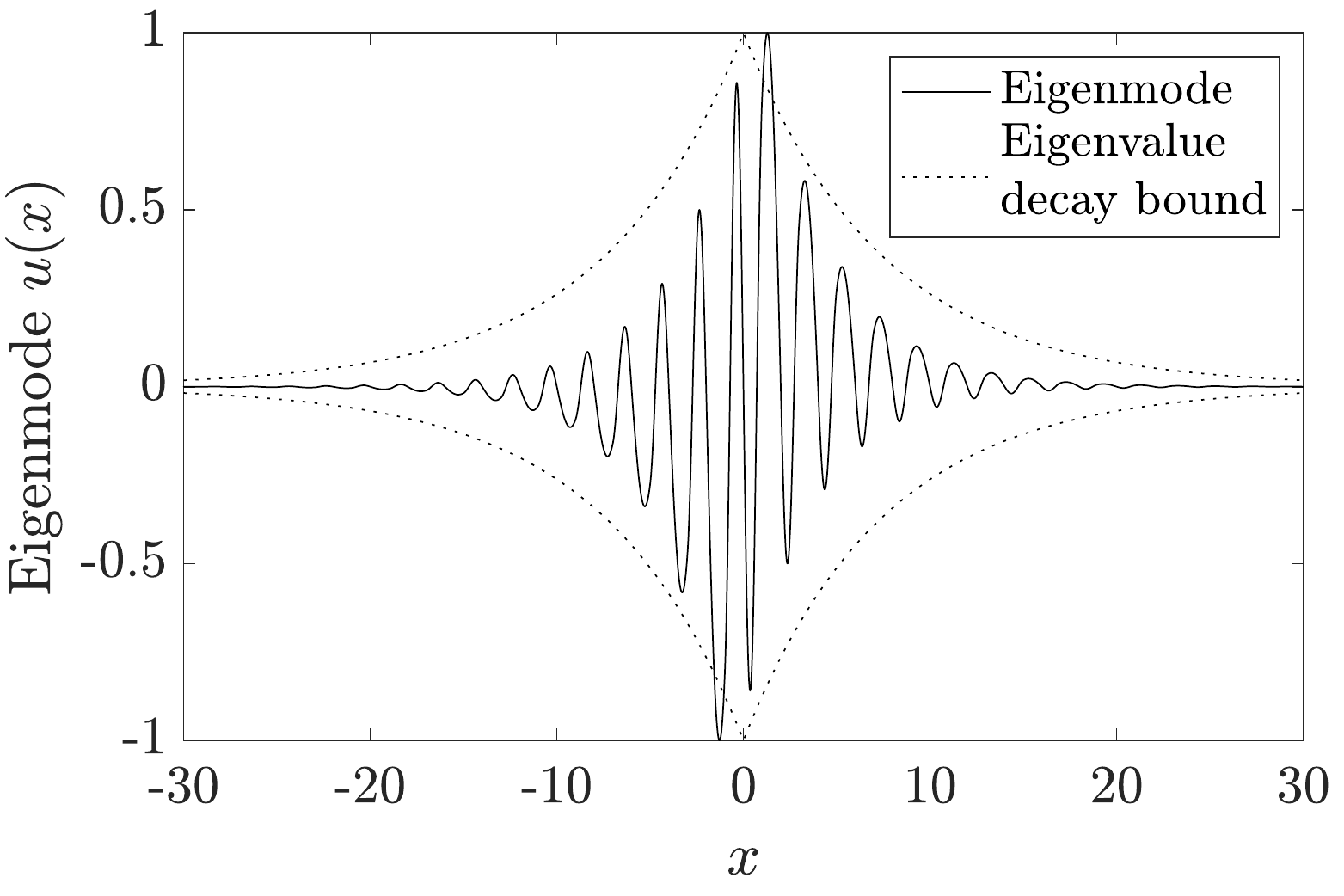}
    \caption{$\omega=4.383$.}
    \end{subfigure}
    \caption{The reflected periodic structure supports localised modes supported that decay away from the interface. These modes are simulated for a finite material with 55 pieces of material on either side of the interface, with material contrast parameter $r=2$. The modes decay exponentially quickly, at a rate that can be bound from above using the eigenvalues of the transfer matrices (shown with a dotted line).}
    \label{fig:Periodicmodes}
\end{figure}

\begin{remark}[The Mandelbrot set]
A natural question to ask, when faced with materials that yield second-order recursion relations like the one in \Cref{lem:recursion}, is whether any particular choices of reflected recursive materials lead to relations that have been studied previously. The most notable, and widely studied, of which is the quadratic recursion relation that defines the Mandelbrot set: the Mandelbrot set is given by $c\in\mathbb{C}$ such that the recursion relation $z_{N+1}=z_N^2+c$ with $z_1=c$ remains bounded for all $N\in\N$. If we have a periodic medium which is given by $P_{N+1}=P_N P_N$, then it holds that $T_{P_{N+1}}=T_{P_N}^2$ and taking the trace of this gives that $y_N=\tr(T_{P_N})$ satisfies
\begin{equation}
    y_N=y_{N-1}^2-2\det(T_{P_N})=y_{N-1}^2-2.
\end{equation}
Since $c=-2$ is in the Mandelbrot set, if $\omega$ is such that the starting material satisfies $\tr(T_{P_1}(\omega))=-2$, then $\omega$ is in a spectral band of the material (\emph{i.e.} \eqref{cond1} from \Cref{thm:spectrum} is not satisfied). For example, if $P_1$ is a homogeneous material of unit length, then we have that $\tr(T_{P_1}(\omega))=2\cos(\omega)$, from which we can deduce that $\omega=(2n+1)\pi$ is in a spectral band, for all $n\in\N$. While invoking properties of the Mandelbrot set to prove that a homogeneous material is transparent at certain frequencies is clearly a wasted effort, this suggests a valuable path of future investigation: using the known properties of widely studied fractals (see \emph{e.g.} \cite{barnsley2012fractals} for a summary) to reveal properties of novel quasicrystalline media without the need for extensive new analysis.
\end{remark}

\section{Robustness of quasicrystalline waveguides} \label{sec:robust}

An important consideration when designing waveguides that are based on making perturbations to induce edge modes is their stability with respect to imperfections. When devices are manufactured, errors and defects will inevitably be introduced and it is important to understand how these undesirable perturbations interact with the one that was made to create the interface. Thus, for the waveguide to be of practical use, it is important that it retains its waveguiding properties when random errors are introduced to the structure. 

Most metamaterial devices inherently have some robustness, in the sense that their properties depend continuously on small imperfections \cite{davies2021robustness}. However, a significant goal in the field is to optimise this robustness. The main breakthrough in this direction has been the discovery of so-called \emph{topologically protected} edge modes, which leverage topological properties of the underlying periodic media to create modes which experience enhanced robustness to imperfections \cite{khanikaev2013photonic, rechtsman2013photonic}. Previous studies have successfully defined topological indices associated to quasicrystalline media, such as the integrated density of states (IDS) \cite{apigo2018topological, xia2020topological}, and \cite{kraus2012equivalence} showed how these ideas could be extended to discrete Fibonacci systems. However, it is not generally clear how to link these properties to the robustness of edge modes induced by perturbations such as reflections.

Some studies have pointed to quasicrystalline media having appealing robustness properties \cite{marti2021edge, marti2022high, liu2022topological, ni2019observation}. A possible (and very approximate) explanation of this is that since quasicrystals already appear to be `disordered' on short scales, and the underlying order is only clear when you examine the long-range patterns in the structure, making local perturbations is less disruptive to its properties than in a periodic medium. This robustness was shown qualitatively for the specific case of reflection-induced edge modes (in arrays of point scatterers on elastic plates) by \cite{marti2021edge}. We can likewise examine this hypothesis in the setting of the reflection-induced waveguides studied here.

Our robustness study is based on perturbing the values of the material parameters $c$. Both the quasicrystalline Fibonacci example presented in \Cref{sec:fibonacci} and the periodic example presented in \Cref{sec:periodic} are based on piecewise constant building blocks of unit length. We introduce random imperfections to the structure by generating independent random variables $X_n\sim N(0,\sigma^2)$, for $n\in\mathbb{Z}$, and making perturbations on each constant building block according to
\begin{equation} \label{eq:perturb}
    c(x)\mapsto (1+X_n)c(x), \quad x\in [n,n+1), n\in\mathbb{Z}.
\end{equation}
We additionally enforce that $c>0$ by taking $\max\{c,\epsilon\}$ for some small $\epsilon$ (for the simulations presented here, $\epsilon=10^{-3}$ is used). Parts of the spectra, perturbed according to \eqref{eq:perturb}, are shown in \Cref{fig:robustness} for several different geometries. In each case, increasing perturbation size $\sigma$ is shown on the horizontal axis with the eigenfrequencies of the localised modes shown with red diamonds. In each case, the profile of the unperturbed ($\sigma=0$) localised mode is shown inset. The same finite-difference approximations that were used for Figures~\ref{fig:Fibonaccimodes} and~\ref{fig:Periodicmodes} were used for these simulations, with step size $h=0.01$.

\begin{figure}
    \centering
    
    \begin{subfigure}{0.45\linewidth}
    \begin{tikzpicture}
        \node[inner sep=0pt] (Fib) at (0,0)
    {\includegraphics[width=\textwidth]{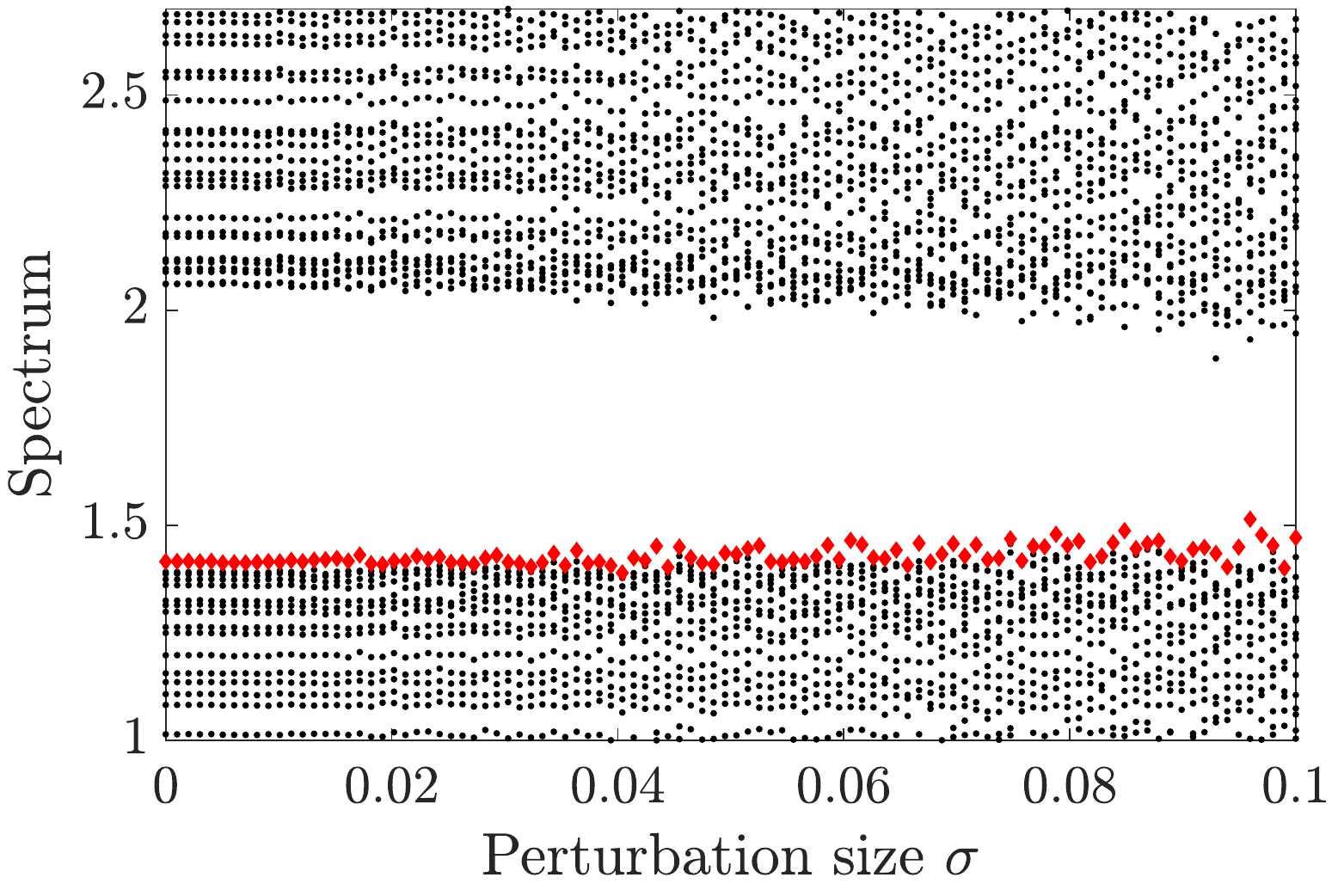}};
        \node[inner sep=0pt] (Fib) at (2.4,1.7)
    {\includegraphics[width=.2\textwidth]{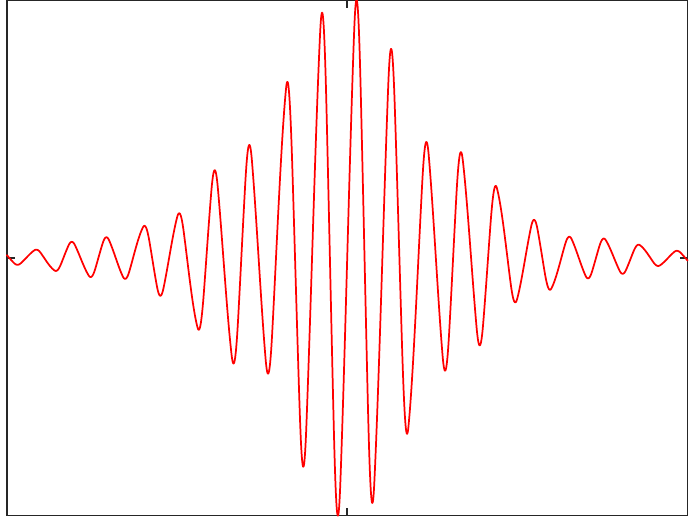}};
    \end{tikzpicture}
    \caption{Fibonacci.}
    \end{subfigure}
    \hspace{0.2cm}
    \begin{subfigure}{0.45\linewidth}
    \begin{tikzpicture}
        \node[inner sep=0pt] (Fib) at (0,0)
    {\includegraphics[width=\textwidth]{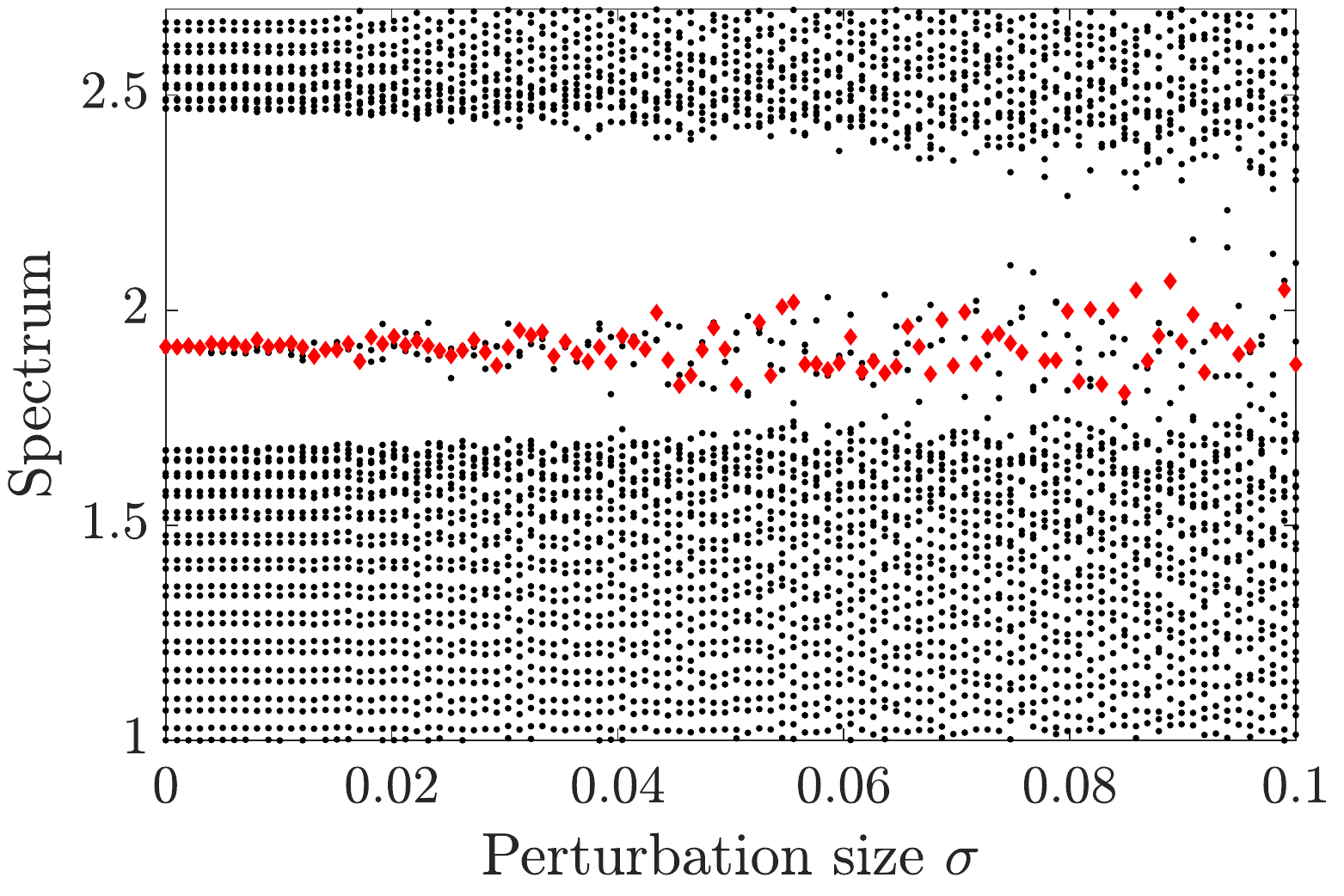}};
        \node[inner sep=0pt] (Fib) at (2.4,-0.9)
    {\includegraphics[width=.2\textwidth]{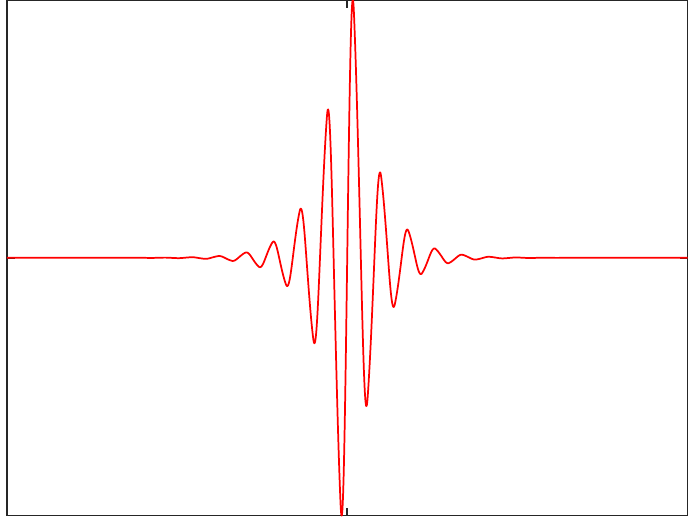}};
    \end{tikzpicture}
    \caption{Periodic.}
    \end{subfigure}
    
    \vspace{0.2cm}
    
    \begin{subfigure}{0.45\linewidth}
    \begin{tikzpicture}
        \node[inner sep=0pt] (Fib) at (0,0)
    {\includegraphics[width=\textwidth]{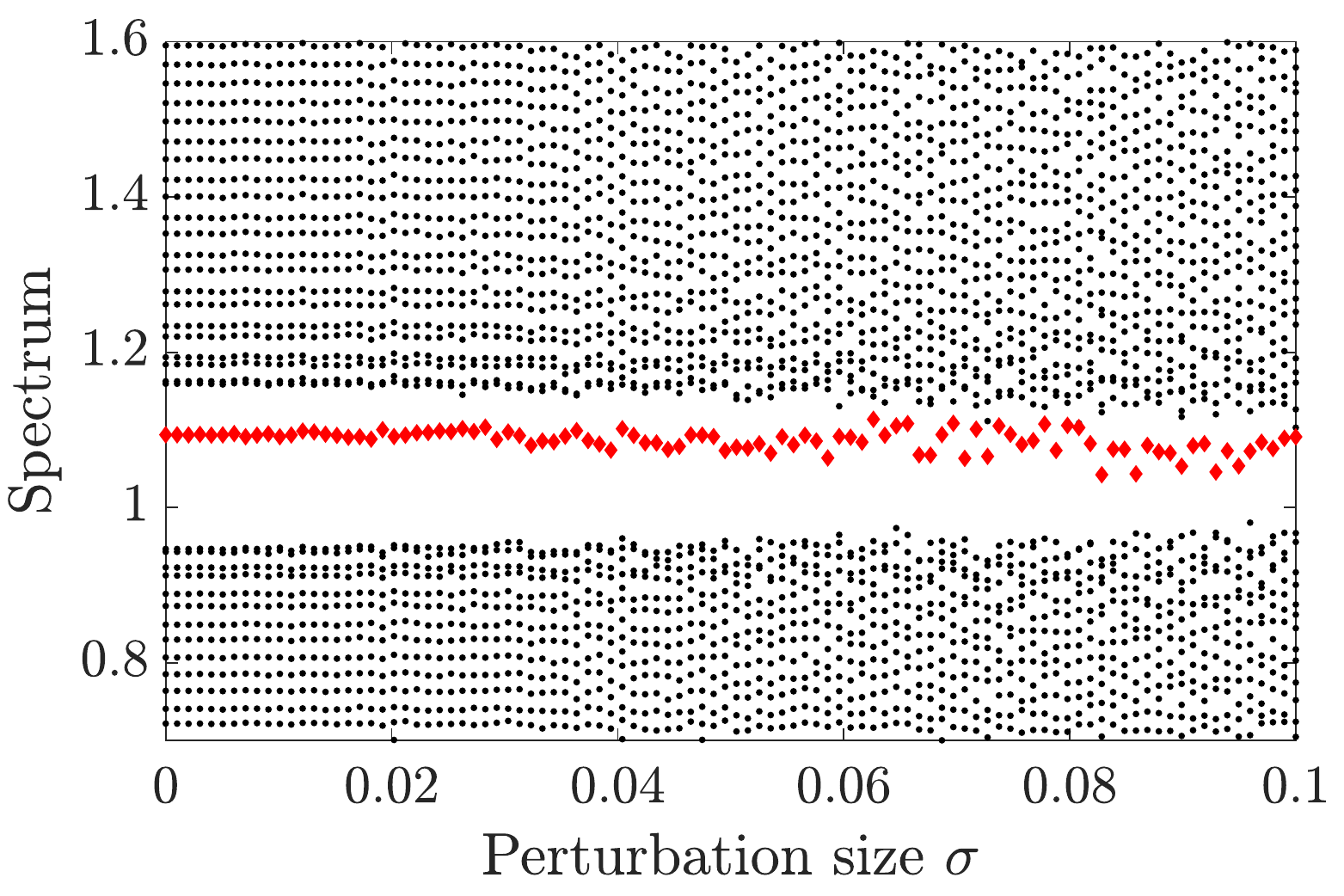}};
        \node[inner sep=0pt] (Fib) at (2.4,1.5)
    {\includegraphics[width=.2\textwidth]{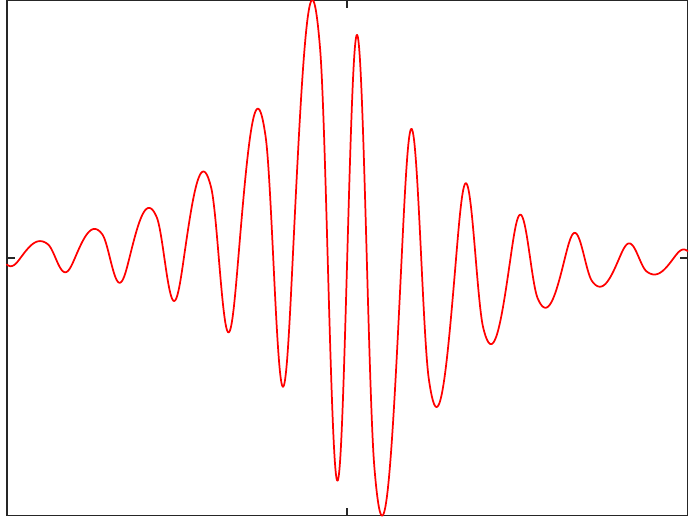}};
    \end{tikzpicture}
    \caption{Stretched periodic.}
    \end{subfigure}
    
    \caption{A comparison of the robustness of the reflection-induced edge modes in Fibonacci and periodic media. We repeatedly compute the spectrum for randomly perturbed media, for 100 different perturbations sizes $\sigma$ that are evenly spaced between 0 and 0.1, with $\sigma$ being the standard deviation of the normally distributed perturbations $N(0,\sigma^2)$. Each point corresponds to an element of the spectrum, with the red diamonds indicating modes that are localised (the unperturbed localised mode is shown inset). In each case, we simulate materials composed of 110 units (55 either side of the interface) with material contrast $r=2$.}
    \label{fig:robustness}
\end{figure}

The initial conclusion to draw from \Cref{fig:robustness}(a) is that the Fibonacci material exhibits strong robustness with respect to imperfections. Even for large perturbation sizes, the localised mode persists and its eigenfrequency appears to remain stable. This is particularly striking when compared with \Cref{fig:robustness}(b), which is the equivalent plot for the periodic material, as sketched in \Cref{fig:Periodic_sketch}. The eigenfrequency is greatly affected by the random perturbations and is mixed up with other localised modes that are created by the perturbations. However, another important difference between these two materials is the rate at which the unperturbed eigenmodes decay (the modes are shown inset in \Cref{fig:robustness}). The reflection induced edge mode in the periodic material decays significantly faster than the mode in the Fibonacci quasicrystal. If we construct a `stretched' version of the periodic material, with $P_1=AAAABBBB$ instead of $P_1=AB$, then the localised mode is similarly stretched by a factor of (approximately) four, such that its decay rate is similar to the mode in the Fibonacci quasicrystal. A robustness analysis for this configuration is shown in \Cref{fig:robustness}(c), where we see that the eigenfrequency of the localised mode is much more robust, with a degree of stability that is comparable to that of the Fibonacci quasicrystal.

The conclusion to draw from this brief robustness study is that while quasicrystalline waveguides may have some beneficial robustness properties, care is needed to make comparisons with \emph{e.g.} periodic waveguides. For example, the robustness benefits may not be clear if decay rates are matched. Whether the comparison should be made between materials with the same sized building blocks or between those that exhibit similar decay rates may depend on the specific application that one has in mind.

\section{Concluding remarks}

This work has established general results for characterising reflection-induced localised modes in materials formed by recursive tiling rules. We have demonstrated the power of this framework by using it to study these modes in the famous Fibonacci quasicrystal. This revealed the rich fractal structure of spectral gaps and demonstrated that adding a reflection creates exponentially localised modes with eigenfrequencies within these spectral gaps. We also showed that these symmetry-induced quasicrystalline waveguides may enjoy enhanced robustness properties, as observed by previous studies \cite{marti2021edge, marti2022high, liu2022topological, ni2019observation}, but that the comparison with the equivalent periodic waveguide requires care and depends on the specific application in mind.

One of the main challenges with studying quasicrystalline media is that it is difficult to formulate general statements and methods. Many previous studies have used methods that rely on the specific properties of the tiling rule, so yield results that are specifically tailored to the quasicrystal in question. This is also true of the results in \Cref{sec:fibonacci} of this work, for example, many of which rely on the recursion relation from \Cref{lem:recursion}. However, the general framework established in \Cref{sec:general} provides a starting point for more systematic studies, which could reveal the properties of the vast collection of unexplored quasicrystals.

A notable feature of the approach developed in this work is that the theoretical results are computationally expensive to implement numerically. The results are iterative in nature, as they rely on the recursive tiling rules that define the materials, and require many iterations to be computed in order to understand the materials' properties. This is common to many studies \cite{gei2010wave, gei2018waves, gei2020phononic, guenneau2008acoustic, hassouani2006surface}. However, as demonstrated in \Cref{sec:periodic} of this work, the application of this theory to periodic materials is significantly more straightforward. It might be possible to take advantage of the fact that quasicrystals are projections of higher-dimensional periodic structures to perform this analysis more concisely. Similar \emph{superspace lifting} approaches have been used previously to reveal new insight into the properties of quasicrystals \cite{johnson2008quasicrystal, bouchitte2010homogenization, wellander2018two}, however relating the eigenmodes of the quasicrystal and the lifted periodic system is non-trivial in general.

The results of this paper, and related works, are reducing the barriers to using quasicrystals in waveguiding devices. Thanks to the convenience of Floquet-Bloch analysis, periodic waveguides are ubiquitous in wave physics. Conversely, the wave transmission properties of quasicrystals can be fiendishly difficult to understand, so their use in applications has been limited. Given their exciting properties (such as large spectral gaps and robustness), the rapidly growing sophistication of techniques for studying quasicrystalline wave systems has the potential to unlock a new class of powerful wave devices based on quasicrystals.

\section*{Acknowledgements}

BD and RVC were funded by the EC-H2020 FETOpen project BOHEME under grant agreement 863179. The authors would like to thank Jordan Docking for pointing out the periodic orbit that was used in \Cref{lem:unboundedsequences} as well as Ian Hooper and Henry Putley for their insightful comments that were valuable in refining the robustness analysis presented in \Cref{sec:robust}. The code used for the numerical examples presented in this article is available at \url{https://doi.org/10.5281/zenodo.7006818}.

\section*{Competing interests}

The authors declare no competing interests.

\bibliographystyle{unsrt}
\bibliography{QPreflection}{}

\begin{thebibliography}{10}

\bibitem{asboth2016short}
J.~K. Asb{\'o}th, L.~Oroszl{\'a}ny, and A.~P{\'a}lyi.
\newblock {\em A Short Course on Topological Insulators}, volume 919 of {\em
  Lecture Notes in Physics}.
\newblock Springer, 2016.

\bibitem{hoefer2011defect}
M.~A. Hoefer and M.~I. Weinstein.
\newblock Defect modes and homogenization of periodic {Schr{\"o}dinger}
  operators.
\newblock {\em SIAM J. Math. Anal.}, 43(2):971--996, 2011.

\bibitem{craster2022ssh}
R.~V. Craster and B.~Davies.
\newblock Asymptotic characterisation of localised defect modes:
  {Su}-{Schrieffer}-{Heeger} and related models.
\newblock {\em arXiv preprint arXiv:2202.07324}, 2022.

\bibitem{kuchment1993floquet}
P.~Kuchment.
\newblock {\em {Floquet Theory for Partial Differential Equations}}.
\newblock Number~60 in Operator Theory: Advances and Applications.
  {Birkh\"auser Verlag}, Basel, 1993.

\bibitem{khanikaev2013photonic}
A.~B. Khanikaev, S.~H. Mousavi, W.-K. Tse, M.~Kargarian, A.~H. MacDonald, and
  G.~Shvets.
\newblock Photonic topological insulators.
\newblock {\em Nat. Mater.}, 12(3):233--239, 2013.

\bibitem{rechtsman2013photonic}
M.~C. Rechtsman, J.~M. Zeuner, Y.~Plotnik, Y.~Lumer, D.~Podolsky, F.~Dreisow,
  S.~Nolte, M.~Segev, and A.~Szameit.
\newblock Photonic {Floquet} topological insulators.
\newblock {\em Nature}, 496(7444):196--200, 2013.

\bibitem{makwana2018geometrically}
M.~P. Makwana and R.~V. Craster.
\newblock Geometrically navigating topological plate modes around gentle and
  sharp bends.
\newblock {\em Phys. Rev. B}, 98(18):184105, 2018.

\bibitem{zolla1998remarkable}
F.~Zolla, D.~Felbacq, and B.~Guizal.
\newblock A remarkable diffractive property of photonic quasi-crystals.
\newblock {\em Opt. Commun.}, 148(1-3):6--10, 1998.

\bibitem{marti2022high}
M.~Mart{\'\i}-Sabat{\'e}, S.~Guenneau, and D.~Torrent.
\newblock High-quality resonances in quasi-periodic clusters of scatterers for
  flexural waves.
\newblock {\em arXiv preprint arXiv:2205.15038}, 2022.

\bibitem{marti2021edge}
M.~Mart{\'\i}-Sabat{\'e} and D.~Torrent.
\newblock Edge modes for flexural waves in quasi-periodic linear arrays of
  scatterers.
\newblock {\em APL Mater.}, 9(8):081107, 2021.

\bibitem{liu2022topological}
Y.~Liu, L.~F. Santos, and E.~Prodan.
\newblock Topological gaps in quasiperiodic spin chains: A numerical and
  {K}-theoretic analysis.
\newblock {\em Phys. Rev. B}, 105(3):035115, 2022.

\bibitem{ni2019observation}
X.~Ni, K.~Chen, M.~Weiner, D.~J. Apigo, C.~Prodan, A.~Alu, E.~Prodan, and A.~B.
  Khanikaev.
\newblock Observation of {Hofstadter} butterfly and topological edge states in
  reconfigurable quasi-periodic acoustic crystals.
\newblock {\em Commun. Phys.}, 2(1):1--7, 2019.

\bibitem{berger1966undecidability}
R.~Berger.
\newblock {\em The undecidability of the domino problem}.
\newblock Number~66 in Memoirs of the American Mathematical Society. American
  Mathematical Soc., Providence, RI, 1966.

\bibitem{wang1961proving}
H.~Wang.
\newblock Proving theorems by pattern recognition—ii.
\newblock {\em Bell Syst. Tech. J.}, 40(1):1--41, 1961.

\bibitem{penrose1974role}
R.~Penrose.
\newblock The role of aesthetics in pure and applied mathematical research.
\newblock {\em Bull. Inst. Math. Appl.}, 10:266--271, 1974.

\bibitem{shechtman1984metallic}
D.~Shechtman, I.~Blech, D.~Gratias, and J.~W. Cahn.
\newblock Metallic phase with long-range orientational order and no
  translational symmetry.
\newblock {\em Phys. Rev. Lett.}, 53(20):1951, 1984.

\bibitem{kohmoto1983localization}
M.~Kohmoto, L.~P. Kadanoff, and C.~Tang.
\newblock Localization problem in one dimension: Mapping and escape.
\newblock {\em Phys. Rev. Lett.}, 50(23):1870, 1983.

\bibitem{kohmoto1987critical}
M.~Kohmoto, B.~Sutherland, and C.~Tang.
\newblock Critical wave functions and a {Cantor}-set spectrum of a
  one-dimensional quasicrystal model.
\newblock {\em Phys. Rev. B}, 35(3):1020, 1987.

\bibitem{hassouani2006surface}
Y.~El~Hassouani, H.~Aynaou, E.~H. El~Boudouti, B.~Djafari-Rouhani, A.~Akjouj,
  and V.~R. Velasco.
\newblock Surface electromagnetic waves in {Fibonacci} superlattices:
  Theoretical and experimental results.
\newblock {\em Phys. Rev. B}, 74(3):035314, 2006.

\bibitem{guenneau2008acoustic}
S.~Guenneau, A.~B. Movchan, N.~V. Movchan, and J.~Trebicki.
\newblock Acoustic stop bands in almost-periodic and weakly randomized
  stratified media: perturbation analysis.
\newblock {\em Acta Mech. Sinica}, 24(5):549--556, 2008.

\bibitem{gei2010wave}
M.~Gei.
\newblock Wave propagation in quasiperiodic structures: stop/pass band
  distribution and prestress effects.
\newblock {\em Int. J. Solids Struct.}, 47(22-23):3067--3075, 2010.

\bibitem{gei2018waves}
L.~Morini and M.~Gei.
\newblock Waves in one-dimensional quasicrystalline structures: Dynamical trace
  mapping, scaling and self-similarity of the spectrum.
\newblock {\em J. Mech. Phys. Solids}, 119:83--103, 2018.

\bibitem{gei2020phononic}
M.~Gei, Z.~Chen, F.~Bosi, and L.~Morini.
\newblock Phononic canonical quasicrystalline waveguides.
\newblock {\em Appl. Phys. Lett.}, 116(24):241903, 2020.

\bibitem{spurrier2020kane}
S.~Spurrier and N.~R. Cooper.
\newblock Kane-{Mele} with a twist: Quasicrystalline higher-order topological
  insulators with fractional mass kinks.
\newblock {\em Phys. Rev. Res.}, 2(3):033071, 2020.

\bibitem{kolar1993quasicrystals}
M.~Kol\'a{\v r}.
\newblock New class of one-dimensional quasicrystals.
\newblock {\em Phys. Rev. B}, 47(9):5489--5492, 1993.

\bibitem{SANCHEZSOTO2012191}
L.~L. Sánchez-Soto, J.~J. Monzón, A.~G. Barriuso, and J.~F. Cariñena.
\newblock The transfer matrix: a geometrical perspective.
\newblock {\em Phys. Rep.}, 513(4):191--227, 2012.

\bibitem{felbacq1998limit}
D.~Felbacq, B.~Guizal, and F.~Zolla.
\newblock Limit analysis of the diffraction of a plane wave by a
  one-dimensional periodic medium.
\newblock {\em J. Math. Phys.}, 39(9):4604--4607, 1998.

\bibitem{zolla1998wave}
D.~Felbacq, B.~Guizal, and F.~Zolla.
\newblock Wave propagation in one-dimensional photonic crystals.
\newblock {\em Opt. Commun.}, 152(1-3):119--126, 1998.

\bibitem{johnson2008quasicrystal}
A.~W. Rodriguez, A.~P. McCauley, Y.~Avniel, and S.~G. Johnson.
\newblock Computation and visualization of photonic quasicrystal spectra via
  {Bloch’s} theorem.
\newblock {\em Phys. Rev. B}, 77(10):104201, 2008.

\bibitem{bouchitte2010homogenization}
G.~Bouchitt{\'e}, S.~Guenneau, and F.~Zolla.
\newblock Homogenization of dielectric photonic quasi crystals.
\newblock {\em Multiscale Model. Sim.}, 8(5):1862--1881, 2010.

\bibitem{wellander2018two}
N.~Wellander, S.~Guenneau, and E.~Cherkaev.
\newblock Two-scale cut-and-projection convergence; homogenization of
  quasiperiodic structures.
\newblock {\em Math. Meth. Appl. Sci.}, 41(3):1101--1106, 2018.

\bibitem{apigo2019observation}
D.~J. Apigo, W.~Cheng, K.~F. Dobiszewski, E.~Prodan, and C.~Prodan.
\newblock Observation of topological edge modes in a quasiperiodic acoustic
  waveguide.
\newblock {\em Phys. Rev. Lett.}, 122(9):095501, 2019.

\bibitem{xia2020topological}
Y.~Xia, A.~Erturk, and M.~Ruzzene.
\newblock Topological edge states in quasiperiodic locally resonant
  metastructures.
\newblock {\em Phys. Rev. Appl.}, 13(1):014023, 2020.

\bibitem{kraus2012topological}
Y.~E. Kraus, Y.~Lahini, Z.~Ringel, M.~Verbin, and O.~Zilberberg.
\newblock Topological states and adiabatic pumping in quasicrystals.
\newblock {\em Phys. Rev. Lett.}, 109(10):106402, 2012.

\bibitem{bandres2016topological}
M.~A. Bandres, M.~C. Rechtsman, and M.~Segev.
\newblock Topological photonic quasicrystals: Fractal topological spectrum and
  protected transport.
\newblock {\em Phys. Rev. X}, 6(1):011016, 2016.

\bibitem{pal2019topological}
R.~K. Pal, M.~I.~N. Rosa, and M.~Ruzzene.
\newblock Topological bands and localized vibration modes in quasiperiodic
  beams.
\newblock {\em New J. Phys.}, 21(9):093017, 2019.

\bibitem{pu2022topological}
X.~Pu, A.~Palermo, and A.~Marzani.
\newblock Topological edge states of quasiperiodic elastic metasurfaces.
\newblock {\em Mech. Syst. Signal Pr.}, 181(109478), 2022.

\bibitem{apigo2018topological}
D.~J. Apigo, K.~Qian, C.~Prodan, and E.~Prodan.
\newblock Topological edge modes by smart patterning.
\newblock {\em Phys. Rev. Mater.}, 2(12):124203, 2018.

\bibitem{hamilton2021effective}
J.~K. Hamilton, M.~Camacho, R.~R. Boix, I.~R. Hooper, and C.~R. Lawrence.
\newblock Effective-periodicity effects in {Fibonacci} slot arrays.
\newblock {\em Phys. Rev. B}, 104(24):L241412, 2021.

\bibitem{furstenberg1963noncommuting}
H.~Furstenberg.
\newblock Noncommuting random products.
\newblock {\em Trans. Am. Math. Soc.}, 108(3):377--428, 1963.

\bibitem{zhou2019topological}
D.~Zhou, L.~Zhang, and X.~Mao.
\newblock Topological boundary floppy modes in quasicrystals.
\newblock {\em Phys. Rev. X}, 9(2):021054, 2019.

\bibitem{ammari2020robust}
H.~Ammari, B.~Davies, and E.~O. Hiltunen.
\newblock Robust edge modes in dislocated systems of subwavelength resonators.
\newblock {\em J. Lond. Math. Soc.}, doi:10.1112/jlms.12619, 2022.

\bibitem{drouot2020defect}
A.~Drouot, C.~L. Fefferman, and M.~I. Weinstein.
\newblock Defect modes for dislocated periodic media.
\newblock {\em Commun. Math. Phys.}, 377(3):1637--1680, 2020.

\bibitem{hsu2016bound}
C.~W. Hsu, B.~Zhen, D.~Stone, J.~D. Joannopoulos, and M.~Solja{\v{c}}i{\'c}.
\newblock Bound states in the continuum.
\newblock {\em Nat. Rev. Mater.}, 1(9):1--13, 2016.

\bibitem{ammari2021boundstates}
H.~Ammari, B.~Davies, E.~O. Hiltunen, H.~Lee, and S.~Yu.
\newblock Bound states in the continuum and {Fano} resonances in subwavelength
  resonator arrays.
\newblock {\em J. Math. Phys.}, 62(10):101506, 2021.

\bibitem{SchnitzerPorter}
O.~Schnitzer and R.~Porter.
\newblock Acoustics of a partially partitioned narrow slit connected to a
  half-plane: Case study for exponential quasi-bound states in the continuum
  and their resonant excitation.
\newblock {\em SIAM J. Appl. Math.}, 82(4):1387--1410, 2022.

\bibitem{deponti2020graded}
J.~M. De~Ponti, A.~Colombi, R.~Ardito, F.~Braghin, A.~Corigliano, and R.~V.
  Craster.
\newblock Graded elastic metasurface for enhanced energy harvesting.
\newblock {\em New J. Phys.}, 22(1):013013, 2020.

\bibitem{davies2021robustness}
B.~Davies and L.~Herren.
\newblock Robustness of subwavelength devices: a case study of cochlea-inspired
  rainbow sensors.
\newblock {\em Proc. R. Soc. A}, 478(2262):20210765, 2022.

\bibitem{borland1963nature}
R.~E. Borland.
\newblock The nature of the electronic states in disordered one-dimensional
  systems.
\newblock {\em Proc. R. Soc. Lond. A: Math. Phys.}, 274(1359):529--545, 1963.

\bibitem{carmona1982exponential}
R.~Carmona.
\newblock Exponential localization in one dimensional disordered systems.
\newblock {\em Duke Math. J.}, 49(1):191--213, 1982.

\bibitem{comtet2013lyapunov}
A.~Comtet, C.~Texier, and Y.~Tourigny.
\newblock Lyapunov exponents, one-dimensional {Anderson} localization and
  products of random matrices.
\newblock {\em J. Phys. A: Math. Theor.}, 46(25):254003, 2013.

\bibitem{scales1997lyapunov}
J.~A. Scales and E.~S. Van~Vleck.
\newblock Lyapunov exponents and localization in randomly layered media.
\newblock {\em J. Comput. Phys.}, 133(1):27--42, 1997.

\bibitem{barnsley2012fractals}
M.~F. Barnsley.
\newblock {\em Fractals Everywhere}.
\newblock Academic Press, 3 edition, 2012.

\bibitem{kraus2012equivalence}
Y.~E. Kraus and O.~Zilberberg.
\newblock Topological equivalence between the {Fibonacci} quasicrystal and the
  {Harper} model.
\newblock {\em Phys. Rev. Lett.}, 109(11):116404, 2012.

\end{thebibliography}
\end{document}